\documentclass[journal]{IEEEtran}
\usepackage{amsmath,amsfonts}
\usepackage{algorithmic}
\usepackage{algorithm}
\usepackage{array}
\usepackage[caption=false,font=footnotesize,labelfont=sf,textfont=sf]{subfig}
\usepackage[colorlinks=true, linkcolor=black, citecolor=black, urlcolor=black]{hyperref}
\usepackage{textcomp}
\usepackage{stfloats}
\usepackage{url}
\usepackage{verbatim}
\usepackage{graphicx}
\usepackage{cite}
\usepackage{bm}
\usepackage{amsthm}
\usepackage{proof}
\usepackage{amssymb}
\usepackage{color}
\DeclareMathOperator{\Tr}{Tr}
\DeclareMathOperator{\rank}{rank}
\newtheorem{remark}{Remark}
\newtheorem{lemma}{Lemma}
\newtheorem{corollary}{Corollary}
\newtheorem{proposition}{Proposition}
\usepackage[font=small,skip=3pt]{caption}
\allowdisplaybreaks
\begin{document}
	
	\title{Near-Field Multi-Cell ISCAP with Extremely Large-Scale Antenna Array }
	
	\author{Yuan Guo,~\IEEEmembership{Member,~IEEE,}
		Yilong~Chen,~\IEEEmembership{Graduate Student Member,~IEEE,}
		Zixiang~Ren,~\IEEEmembership{Graduate Student Member,~IEEE,}
		Derrick~Wing~Kwan~Ng,~\IEEEmembership{Fellow,~IEEE,}
		and~Jie~Xu,~\IEEEmembership{Fellow,~IEEE}
		\IEEEcompsocitemizethanks{\IEEEcompsocthanksitem Yuan Guo, Yilong Chen, Zixiang Ren, and Jie Xu are with the School of Science and Engineering (SSE), the Shenzhen Future Network of Intelligence Institute (FNii-Shenzhen), and the Guangdong Provincial Key Laboratory of Future Networks of Intelligence, The Chinese University of Hong Kong, Shenzhen, Guangdong 518172, China (e-mail: guoyuan@cuhk.edu.cn, lilongchen@link.cuhk.edu.cn, rzx66@mail.ustc.edu.cn, xujie@cuhk.edu.cn).\\	
		\indent  Derrick~Wing~Kwan~Ng is with the School of Electrical Engineering and Telecommunications, University of New South Wales, Sydney, NSW 2052, Australia (e-mail: w.k.ng@unsw.edu.au).\\
		\indent Jie Xu is the corresponding author.}
	}

	
	\maketitle
	
	\begin{abstract}
		This paper investigates a coordinated multi-cell integrated sensing, communication, and powering (ISCAP) system operating in the electromagnetic near field, where each base station (BS) employs an extremely large-scale antenna array (ELAA) to simultaneously support downlink communication, wireless power transfer (WPT), and environmental sensing. Three categories of communication users (CUs) with different interference cancellation capabilities are considered, and sensing is enabled through a distributed multiple-input multiple-output (MIMO) radar architecture. To address the resulting design challenges, a robust optimization framework is proposed by optimizing the beamforming strategy to maximize the worst-case detection probability over a prescribed sensing region, subject to  per-user signal-to-interference-plus-noise ratio (SINR) constraints and energy harvesting requirements at energy receivers (ERs), while explicitly capturing the uncertainty in ER locations. By leveraging semidefinite relaxation (SDR), the original non-convex problem is reformulated as a convex semidefinite program with a provably tight relaxation. Furthermore, a low-complexity maximum ratio transmission (MRT)-based suboptimal scheme is developed, yielding a closed-form solution in the asymptotic regime as the number of antenna elements approaches infinity. Extensive numerical results reveal the fundamental trade-offs among sensing accuracy, communication reliability, and WPT efficiency. Specifically, our results demonstrate that: i) the proposed SDR- and MRT-based coordinated designs consistently outperform the benchmark scheme without inter-cell coordination; ii) greater ER location uncertainty degrades overall system performance, while the proposed spatially averaged robust framework outperforms conventional worst-case robust WPT designs; iii) receiver-side interference cancellation improves sensing performance only when CUs are positioned near the sensing area; iv) although larger BS power budgets enhance sensing performance, this improvement is limited by stringent false alarm requirements; and v) near-field ISCAP systems leveraging ELAAs provide clear advantages over far-field configurations by enabling fine-grained spatial resolution for joint energy focusing and interference suppression.

	\end{abstract}

	\section{Introduction}
	The advent of sixth-generation (6G) wireless networks promises revolutionary capabilities, enabling advanced applications such as extended reality, autonomous vehicles, and industrial automation, all of which demand ultra-high data rates, ultra-low latency, centimeter-level positioning, and sustainable wireless power delivery for billions of Internet-of-Things (IoT) devices \cite{recommendation2023framework}. However, legacy network architectures, which typically segregate communication, sensing, and power delivery functions, are increasingly inadequate for satisfying these heterogeneous and demanding requirements efficiently.  To this end, integrated sensing, communication, and powering (ISCAP) has emerged as a promising paradigm, establishing a unified framework that simultaneously supports high-throughput data transmission, precise environmental sensing, and reliable power transfer \cite{yilongiscap,yilongofdm,aref2025otsm,ren2025robust}. 
	
	Recent advances have demonstrated that such integration is indeed technically feasible through meticulous co-design at the signal, protocol, and hardware levels \cite{liu2022integrated,hua2023optimal,zhang9540344,clerckx2018fundamentals,ng2013wireless,xu2014multiuser}. In particular, integrated sensing and communication (ISAC) has shown that environmental sensing and data transmission can efficiently share spectrum and infrastructure, enabling novel capabilities such as radar-aided beam alignment and communication-assisted target tracking, thereby enhancing both sensing and communication performance \cite{liu2022integrated,hua2023optimal,zhang9540344}. In parallel, simultaneous wireless information and power transfer (SWIPT) enables the joint transmission of data and power, facilitating battery-less device operation, dense IoT deployments, and long-term network sustainability in scenarios where frequent battery replacement or recharging are impractical \cite{clerckx2018fundamentals,ng2013wireless,xu2014multiuser}. Building on these foundations, ISCAP unifies ISAC and SWIPT in a pragmatic air-interface, in which a co-designed transmit signal concurrently supports all three core functionalities envisioned for upcoming 6G networks  \cite{yilongiscap,yilongofdm,zhou2024integrating,hao2024energy,aref2025otsm,ren2025robust}. In particular, the pioneering work in \cite{yilongiscap} introduced  ISCAP in a multiple-input multiple-output (MIMO) system with unified signal design, which was further extended to a practical orthogonal frequency-division multiplexing (OFDM)-based system in \cite{yilongofdm} and an orthogonal time sequency multiplexing (OTSM)-based millimeter-wave system in \cite{aref2025otsm}. Furthermore, the authors in \cite{zhou2024integrating} investigated ISCAP for multi-user scenarios, while \cite{hao2024energy} proposed an energy-efficient hybrid analog–digital beamforming design with dynamic radio-frequency (RF) chain control to jointly optimize multi-functional performance. While ISCAP offers significant gains in spectral efficiency and energy utilization, through reduced hardware duplication and signaling overhead, it also introduces critical challenges in beam management, interference mitigation, and multi-objective resource allocation, which call for novel infrastructure designs and advanced signal processing techniques to fulfill the diverse requirements of co-existing functions.

	A key enabler for realizing ISCAP is the deployment of extremely large-scale antenna arrays (ELAAs), which offer unprecedented spatial resolution and beamforming precision by operating in the near-field regime. \textcolor{black}{Leveraging these advantages, numerous recent studies have investigated ELAA-enabled near-field communications as a cornerstone technology for 6G wireless networks \cite{liu10716601,liu2023near,liu10944643,wang2023near,Ouyang10639537,ZhaoNear2024,Hua11030222,LWC2022,demarchou2022energy,JSAC2024,zhang2024simultaneous}.} Specifically, the authors in \cite{liu10944643,liu10716601} highlighted the pivotal role of ELAA in 6G infrastructure through  comprehensive surveys of the fundamental principles of near-field communication. Moreover, \cite{wang2023near,ZhaoNear2024,Hua11030222} demonstrated that the distance-dependent nature inherent to the near-field propagation enhances ISAC, by enabling improved trade-offs between sensing and communication compared to conventional far-field designs. \textcolor{black}{Additionally, \cite{LWC2022} proposed a fast near-field beam training scheme for extremely large-scale arrays to efficiently estimate distance–angle parameters, highlighting practical aspects of near-field channel acquisition in single-cell systems.} Furthermore, the work in \cite{demarchou2022energy} highlighted that near-field energy focusing capability enables substantially more efficient wireless power transfer (WPT) compared to the far-field counterparts. \textcolor{black}{The work in \cite{JSAC2024} investigated SWIPT in mixed near- and far-field channels, proposing joint beam scheduling and power allocation for energy delivery in a single-cell scenario.} Moreover, the authors in \cite{zhang2024simultaneous} revealed that, in contrast to far-field SWIPT, a single near-field beam can simultaneously concentrate energy at multiple user locations, reducing beamforming complexity while ensuring joint power and information delivery. More recently, the work  in \cite{ren2025robust} studied near-field ISCAP from a physical-layer security perspective, showing that secure beamforming leveraging ELAAs can effectively balance confidentiality, sensing accuracy, and power transfer.
	
	In practice, near-field propagation has always existed in practical wireless systems, although its impact was historically negligible due to the short Rayleigh distance resulting from the limited apertures of conventional antenna arrays \cite{cong2024near}. With the emergence of ELAAs, however, the Rayleigh distance increases dramatically, extending the near-field region to several hundred meters, rendering near-field propagation effects an important consideration in practical communication system design \cite{liu10716601}. Meanwhile, to accommodate the anticipated demands for ultra-high throughput and connectivity in future 6G networks, their architectures are expected to become significantly denser, featuring multiple ELAA-equipped base stations (BSs) deployed in close proximity \cite{xu2025toward}. This densification creates overlapping near-field regions and strong inter-cell spatial coupling, particularly in device-dense environments. \textcolor{black}{Such dense deployments naturally give rise to cooperative multi-BS operation, where BSs jointly perform communication, sensing, and WPT within shared coverage regions. This coordinated architecture aligns well with the envisioned 6G network evolution toward integrated multi-functional infrastructures, motivating the adoption of a near-field multi-cell ISCAP framework as a physically consistent and practically relevant system model.} Recent studies have examined multi-cell and cell-free systems under such practical near-field conditions from the perspectives of distance-aware channel modeling, interference-aware user scheduling, and system-level evaluation with reconfigurable intelligent surfaces (RISs), with a primary focus on enhancing communication performance (e.g., \cite{xie2024near,PA10892231,yuan2025reconfigurable}). However, this communication-centric perspective falls short of addressing how near-field multi-cell systems can be practically and efficiently designed to support the broader multi-functional capabilities envisioned for 6G networks, such as sensing and WPT. This challenge is further compounded by the nature of energy receivers (ERs), which are typically passive or low-power IoT devices incapable of actively reporting their locations due to energy and hardware constraints \cite{Li9184896,aziz2019battery}. As a result, their positions are only coarsely known and may vary within bounded regions due to mobility or placement uncertainty. Collectively, these challenges underscore the need to explore the full potential of near-field propagation in multi-cell ISCAP systems by tackling key challenges in interference management, multi-functional resource allocation, and robust system design under ER spatial uncertainty, which remains largely unexplored in the current literature.

	Motivated by the aforementioned research gaps, this paper investigates a coordinated multi-cell ISCAP framework in the near-field regime of ELAA-enabled networks, with a specific focus on scenarios involving uncertain ER locations. Our main contributions are summarized as follows:
	\begin{itemize}
		\item We propose an ELAA-enabled near-field ISCAP multi-cell system, which leverages a distributed MIMO radar architecture with signal-level fusion,  accounts for heterogeneous communication users (CUs) with distinct interference cancellation capabilities, and captures spatial uncertainty in ER locations. To this end, a robust optimization framework is developed to maximize the worst-case sensing performance over a designated target area in terms of detection probability, by jointly designing the information beamforming vectors and the covariance matrices of dual-purpose signals for sensing and WPT, subject to per-CU signal-to-interference-plus-noise ratio (SINR) constraints and average energy harvesting requirements at individual ERs.
		
		\item By leveraging the semidefinite relaxation (SDR) technique, we reformulate the original non-convex problem into a convex semidefinite program (SDP) and rigorously prove its relaxation tightness, thereby ensuring the attainment of globally optimal solution. Moreover, aiming to reduce computational complexity, we further develop a low-complexity maximum ratio transmission (MRT)-based sub-optimal design, by fixing beam directions for communication, powering, and sensing, where a closed-form solution is obtained in the asymptotic regime as the number of antenna elements is sufficiently large.

		\item Extensive simulation results reveal the inherent trade-offs among sensing accuracy, communication reliability, and WPT efficiency under varying SINR thresholds, energy demands, and CU/ER deployments. It is shown that larger ER uncertainty regions degrade overall system performance, while the proposed spatially averaged robust framework achieves superior performance  to conventional worst-case robust approaches for handling WPT location uncertainty. The proposed SDR-based optimal and MRT-based suboptimal schemes are shown to consistently outperform non-coordinated baselines, benefiting from coordinated transmission for efficient resource allocation and interference management.
		
		\item Moreover, our results show that receiver-side interference cancellation offers noticeable sensing gains only when CUs are located near the sensing region. Detection performance is more sensitive to ER powering demands than to CU SINR thresholds, which increases with larger BS power budgets but diminishes under stricter false alarm constraints.  Finally, comparisons with far-field setups highlight the distinct advantage of near-field ISCAP, which provides finer spatial resolution for joint energy focusing and interference suppression in dense multi-cell environments.

	\end{itemize}
	
	The rest of this paper is organized as follows. Section \ref{syslabel} describes the system model of the proposed framework. Section \ref{iscapfunctions} introduces the details of ISCAP functionality and formulates the joint optimization problem. Section \ref{optimalsolu} presents the SDR-based optimal and MRT-based sub-optimal solutions. Section \ref{numres} provides numerical results with comprehensive discussions, followed by our conclusions in Section \ref{conlabel}.

	\textit{Notations:} Throughout this paper, vectors and matrices are denoted by boldface lower- and upper-case letters, respectively. $\mathbb{C}^{N \times M}$ and $\mathbb{R}^{N\times M}$ denote the spaces of $N \times M$ matrices with complex and real entries, respectively. For a square matrix $\bm{A}$, $\text{Tr}(\bm{A})$ denotes its trace, and $\bm{A} \succeq \bm{0}$ indicates that $\bm{A}$ is positive semi-definite. For a complex matrix $\bm{A}$ of arbitrary size, $\text{rank}(\bm{A})$, $\bm{A}^{T}$, and $\bm{A}^{H}$ represent its rank, transpose, and complex conjugate transpose, respectively. \(\mathcal{CN}(\bm{x}, \bm{Y})\) denotes the circularly symmetric complex Gaussian (CSCG) distribution with mean vector \( \bm{x} \) and covariance matrix \( \bm{Y} \). The expectation with respect to random variable $x$ is denoted as $\mathbb{E}_x\{\cdot\}$, and $\| \cdot \|$ denotes the Euclidean norm of a vector.  $\bm{A}\odot\bm{B}$ denotes the Hadamard product of two matrices $\bm{A}$ and $\bm{B}$.

	\section{System Model}\label{syslabel}
	\begin{figure*}[t]
		\centering
		\includegraphics[width=0.7\linewidth]{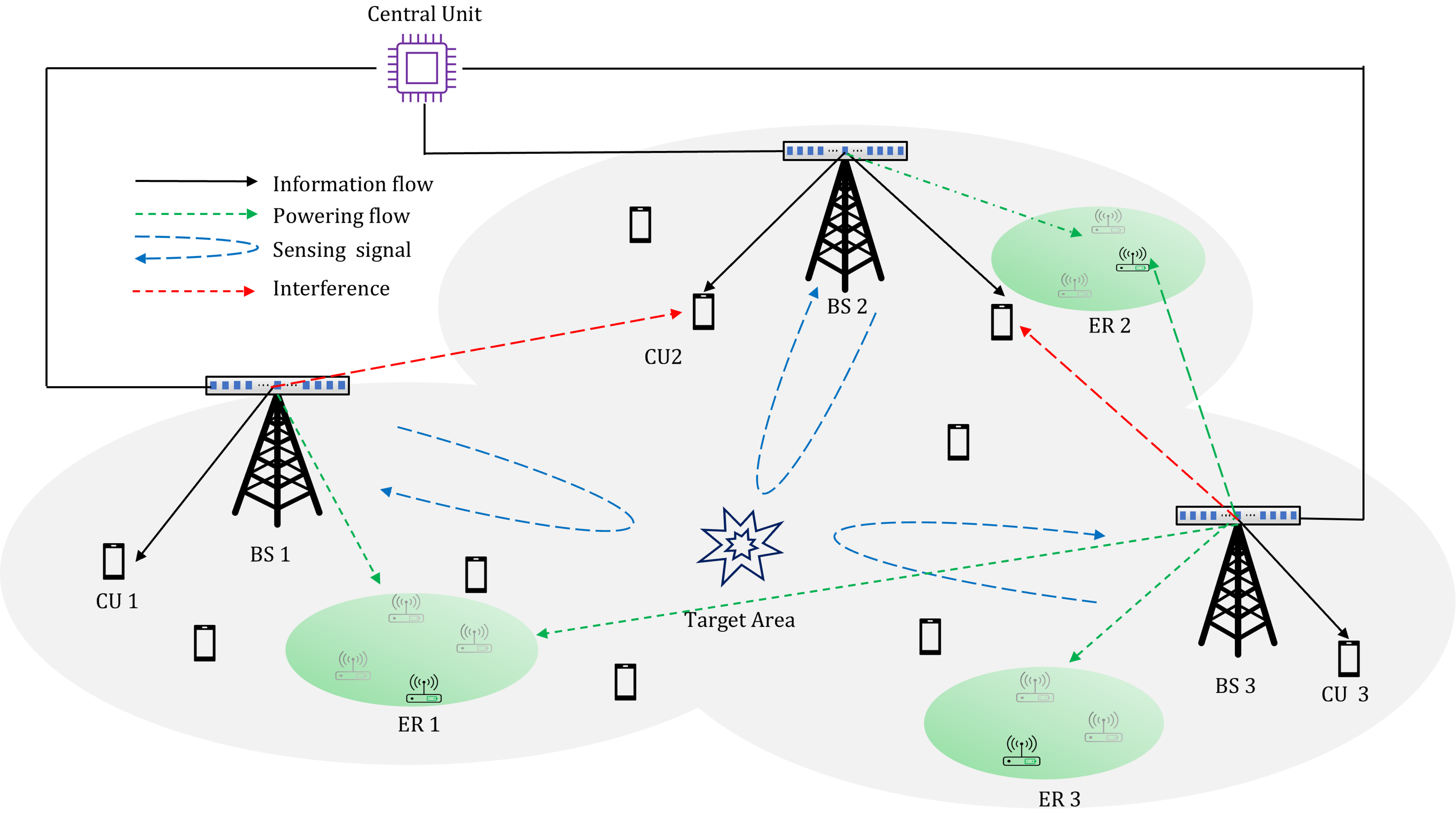}
		\caption{Illustration of the near-field coordinated multi-cell ISCAP system.}
		\label{systemfigures}
	\end{figure*}
	In this section, we present the details of the considered system model. The key mathematical notations used throughout this paper are summarized in Table \ref{notations}.
	\subsection{Network Model}\label{netlabel}
	We consider a coordinated multi-cell downlink network consisting of $K$ BSs with ISCAP capabilities, as illustrated in Fig.~\ref{systemfigures}. Let $\mathcal{K} = \{1, \dotsc, K\}$ denote the set of BSs. Each BS is equipped with an identical uniform linear array (ULA) of $N \gg 1$ antenna elements, with inter-element spacing $d$. Let $\mathcal{N} = \{1, 2, \dotsc, N\}$ denote the set of antenna elements at each BS. The Cartesian coordinate of the $n$-th antenna element at the $k$-th BS is denoted by $\mathbf{q}_{n,k} = [x_{n,k}, y_{n,k}]^T$, $\forall k \in \mathcal{K}$, $n \in \mathcal{N}$. \textcolor{black}{To enable cooperative operation, the coordinated BSs are interconnected via high-speed backhaul links that facilitate control signaling and timing synchronization \cite{cheng2024optimal,Jungnickel4726071,Nigam6928420}.}
	
		\begin{table*}[t]
		\centering
		\caption{Summary of Notations}
        \label{notations}
        \renewcommand{\arraystretch}{1.5}
        \scalebox{0.9}{%
            \begin{tabular}{|l|l||l|l|}
                \hline
                \textbf{Notation} & \textbf{Description} & \textbf{Notation} & \textbf{Description} \\
                \hline
                $K$ & Number of BSs & $\mathbf{q}_{n,k}=(x_{n,k},y_{n,k})$ & Coordinates of the $n$-th antenna element at BS $k$ \\
                \hline
                $\mathcal{K}=\{1,...,K\}$ & Set of coordinated BSs & $\mathcal{C}=\{c_1,...,c_K\}$ & Set of CUs \\
                \hline
                $\mathcal{S}=\{s_1,...,s_M\}$ & Set of sensing sampling points & $\mathcal{E}=\{e_1,...,e_K\}$ & Set of EUs \\
                \hline
                $d$ & Spacing between adjacent ULA elements & $\mathbf{p}_{c_k}=(x_{c_k},y_{c_k})$ & Coordinates of  CU $c_k$ \\
                \hline
                $\lambda$ & Carrier wavelength & $\mathbf{p}_{e_k}=(x_{e_k},y_{e_k})$ &  Coordinates of  ER $e_k$ \\
                \hline
                $N$ & Number of antennas per ULA & $\mathbf{p}_{s_m}=(x_{s_m},y_{s_m})$ & Coordinates of the sensing sampling point $s_m$ \\
                \hline
                $D$ & Effective aperture of ULA, $D=(N-1)d$ & $\bm{h}_{k,c_l}$ & Channel vector between  CU $c_l$ and  BS $k$ \\
                \hline
                $\kappa$ & Wave number of the carrier& $\bm{h}_{k,e_l}$ & Channel vector between  ER $e_l$ and  BS $k$ \\
                \hline
                $M$ & Number of sampling points in $\mathcal{A}_{\rm s}$ & $\bm{h}_{k,s_m}$&Channel vector between sensing point $s_m$ and  BS $k$\\ 
                \hline
                $\mathcal{A}_{\rm s}$	&Designated sensing area& $\mathcal{A}_{e_k}$&Uncertainty region of ER $e_k$\\
                \hline
            \end{tabular}%
		}
	\end{table*}
	We assume  a general multi-user scenario with multiple CUs and ERs. In particular, we consider a representative transmission block in which each BS simultaneously communicates with one CU and transfers power to one ER\footnote{This setup is adopted primarily for illustration. The proposed analytical framework can be readily extended to scenarios where multiple CUs and ERs are served simultaneously by each BS per transmission block.}, while all CUs and ERs are each equipped with a single omnidirectional antenna \cite{yilongofdm,cheng2024optimal}. Let $\mathcal{C} = \{c_1, c_2, \dotsc, c_K\}$ and $\mathcal{E} = \{e_1, e_2, \dotsc, e_K\}$ denote the sets of CUs and ERs served in the considered transmission block, respectively, where $c_k$ and $e_k$ correspond to the CU and ER associated with BS $k$.  The location of each CU, denoted by $\mathbf{p}_{c_k}=[x_{c_k},y_{c_k}]^T$, is assumed to be perfectly known by the network, with accuracy ensured by the active exchange of pilot and data signals with the BS, which enables precise position estimation through standard localization methods \cite{dwivedi2021positioning}.	In contrast, ERs may be passive energy-harvesting devices that do not actively transmit pilots, and thus their positions may not be precisely estimated. Instead, each ER is assumed to lie within an uncertainty region $\mathcal{A}_{e_k}$, reflecting mobility, coarse localization, and limited feedback \cite{Li9184896}. Accordingly, the actual ER position $\mathbf{p}_{e_k}=[x_{e_k},y_{e_k}]^T$ is modeled as uniformly distributed over $\mathcal{A}_{e_k}$\footnote{\textcolor{black}{It is worth noting that the proposed framework is not restricted to any specific geometry and remains applicable to arbitrary spatial configurations of the ERs' uncertainty regions.}}. \textcolor{black}{It is also noted that both CUs and ERs may change their locations across transmission blocks due to mobility. Accordingly, the proposed framework operates on a per-block basis, where the network topology is regarded as quasi-static within each coherent transmission block.} Furthermore, to facilitate target sensing, a designated area $\mathcal{A}_{\mathrm{s}}$ is discretized into $M$ sampling points, $\mathcal{S}=\{s_1,s_2,\dotsc,s_M\}$, where the coordinate of point $s_m$ is given by $\mathbf{p}_{s_m}=[x_{s_m},y_{s_m}]^T$. Such a discretization approach is widely adopted in radar and ISAC literature as it transforms the continuous area into a finite set of candidate locations, thereby enabling tractable optimization while still capturing the spatial characteristics of the target region \cite{cheng2024optimal,stoica2007probing}. Additionally, in the considered coordinated sensing operation, all BSs transmit their processed sensing information to a central unit via backhaul links for joint target detection, with further details discussed in Section \ref{ensensing}. 
	
	Finally,  all CUs, ERs, and sensing sampling points are assumed to be located within the electromagnetic radiative near-field region of the coordinated BSs. Specifically, for any point of interest, $\mathbf{p}$, within the network coverage, its distance to the $n$-th antenna element of the $k$-th BS satisfies $
		\sqrt[3]{\frac{D^4}{8\lambda}} < \|\mathbf{q}_{n,k}-\mathbf{p}\| < \frac{2D^2}{\lambda}$, 
	where $\lambda$ is the carrier wavelength and $D = (N-1)d$ denotes the effective aperture of the ULA \cite{balanis2016antenna}.
	
	\subsection{Near-field Channel Model}
	In the near-field propagation regime, the distances between CUs, ERs, sensing points, and the BSs' ULAs are comparable to the array aperture. As a result, the conventional planar wavefront assumption becomes inaccurate, as the propagation distances vary across antenna elements, leading to element-dependent phase shifts and path loss. To accurately capture this effect, we adopt the spherical wavefront model, based on which the channel vector between BS \(k\) and any point of interest \(\mathbf{p}\in\mathbb{R}^2\) corresponding to a CU, ER, or sensing target, is expressed as \cite{liu10944643,Ouyang10639537,liu2023near,wang2023near}
	\begin{equation}
		\bm{h}_k(\mathbf{p}) = \bm{\ell}_k(\mathbf{p})\odot\bm{v}_k(\mathbf{p}),
	\end{equation}
	where  $\bm{\ell}_k(\mathbf{p})\in \mathbb{C}^{N\times 1}$ accounts for the free-space path loss, and $\bm{v}_k(\mathbf{p})\in\mathbb{C}^{N\times 1}$ is the steering vector of the  ULA at BS \(k\) toward \(\mathbf{p}\), which are given by 
	\begin{align}
		\bm{\ell}_k(\mathbf{p})
		=&\frac{1}{2\kappa}\left[
		\frac{1}{\|\mathbf{q}_{1,k} - \mathbf{p}\|},
		\frac{1}{ \|\mathbf{q}_{2,k} - \mathbf{p}\|},
		\dotsc,
		\frac{1}{ \|\mathbf{q}_{N,k} - \mathbf{p}\|}
		\right]^T\hspace{-3mm},\nonumber
	\end{align}
	and
	\begin{align}
		\bm{v}_{k}(\mathbf{p})
		=&
		\Bigl[
		e^{-j\kappa\|\mathbf{q}_{1,k}-\mathbf{p}\|},
		e^{-j\kappa\|\mathbf{q}_{2,k}-\mathbf{p}\|},
		\ldots,
		e^{-j\kappa\|\mathbf{q}_{N,k}-\mathbf{p}\|}
		\Bigr]^{T},\nonumber
	\end{align}
	respectively, with \(\kappa = 2\pi/\lambda\) and $j=\sqrt{-1}$. It is worth emphasizing that the near-field region in numerous practical wireless scenarios is often characterized by a strong line-of-sight (LoS) component, owing to favorable deployment conditions such as short link distances, open environments, and unobstructed propagation paths~\cite{liu10944643,Ouyang10639537,liu2023near,wang2023near,ZhaoNear2024,liu10716601}. \textcolor{black}{In such environments, scattering is generally weak and multi-path effects are negligible \cite{near10664591,near10841363,near10988518}. Hence, the adopted channel model provides a reliable characterization of near-field ISCAP systems, while the proposed framework remains applicable to more complex multi-scattering propagation conditions. Moreover, for communication links, the channels are assumed to be accurately estimated at the BSs through pilot-based channel estimation within each transmission block, whereas the BS-ER channels are characterized statistically according to their spatial uncertainty regions, as detailed in Section~\ref{wptsection}.}

	For notational clarity, we define the following channel vectors from BS $k\in\mathcal{K}$ to the CUs, ERs, and sensing targets. Specifically,  $\bm{h}_{k,c_l} = \bm{h}_{k}(\mathbf{p}_{c_l})$ denotes the channel vector to  CU $c_l\in\mathcal{C}$; $\bm{h}_{k,e_l} = \bm{h}_{k}(\mathbf{p}_{e_l})$ denotes the channel vector to  ER $e_l\in\mathcal{E}$; and $\bm{h}_{k,s_m} = \bm{h}_{k}(\mathbf{p}_{ s_m})$ denotes the channel vector to  sensing sampling point $s_m\in\mathcal{S}$.

	\subsection{Transmit Signal Model}
	We consider each transmission block consisting of \( L \) symbols, indexed by \( \tau \in \{1, 2, \dotsc, L\} \). At each symbol, a BS transmits an aggregated signal composed of an information-bearing component and a dual-purpose signal adopted for both effective WPT and environmental sensing. In particular, the transmitted signal  by  BS $k$ at symbol \( \tau \) is thus given by
	\begin{equation}
		\bm{s}_k[\tau] = \bm{\omega}_{k}s_{c_k}[\tau] + \bm{s}_{\mathrm{o},k}[\tau],
	\end{equation}
	where \( s_{c_k}[\tau] \in \mathbb{C} \) denotes the zero-mean, unit-variance information symbol intended for  CU $c_k\in\mathcal{C}$; \( \bm{\omega}_k \in \mathbb{C}^{N \times 1} \) is the corresponding beamforming vector; and \( \bm{s}_{\mathrm{o},k}[\tau] \in \mathbb{C}^{N \times 1} \) represents the dual-purpose signal for wireless powering and sensing. In practice, the dual-purpose signal can be \textit{a priori} generated pseudo-random sequences that is modeled as a zero-mean, wide-sense stationary random process with  time-invariant covariance matrix $\bm{R}_{\mathrm{o},k} = \mathbb{E}\{\bm{s}_{\mathrm{o},k}[\tau] \bm{s}_{\mathrm{o},k}^H[\tau]\} \succeq \bm{0}$. Note that both \( \{\bm{\omega}_k\} \) and \( \{\bm{R}_{\mathrm{o},k}\} \) are design variables to be jointly optimized across all coordinated BSs to enhance the performance of the considered ISCAP system. Furthermore, we assume that \( s_{c_k}[\tau] \) is uncorrelated with \( \bm{s}_{\mathrm{o},k}[\tau] \), and the transmitted signals from different BSs, i.e., \( \bm{s}_k[\tau] \) and \( \bm{s}_l[\tau] \), are mutually independent, \(\forall\, k \ne l \). Additionally, the block length \( L \) is assumed to be sufficiently large to enable accurate estimation of the relevant signal statistics \cite{hua2023optimal}. Finally, we  assume that each BS is subject to a block-average transmit power constraint, i.e.,	
	\begin{equation}
		\|\bm{\omega}_k\|^2 + \operatorname{Tr}(\bm{R}_{\mathrm{o},k}) \leq P_{\max},\quad\forall k\in\mathcal{K},
	\end{equation}
	where \( P_{\max} \) denotes the maximum transmit power budget at each BS.

	\section{ISCAP Functionality and Problem Formulation}\label{iscapfunctions}
	In this section, we characterize the core functionalities of the proposed near-field multi-cell ISCAP system, define the corresponding performance metrics, and formulate the associated joint optimization problem.

	\subsection{Information Transfer}\label{cutype}
	The SINR  at each CU is adopted as the key performance metric for downlink communication. Specifically, the signal received at CU $c_k$, denoted by \( y_{c_k}[\tau] \in \mathbb{C} \), is expressed as
	\begin{align}
		y_{k}[\tau] &= \underbrace{\bm{h}_{k,c_k}^H \bm{\omega}_{k} s_{c_k}[\tau]}_{\text{desired signal}} 
		+ \underbrace{\bm{h}_{k,c_k}^H \bm{s}_{\mathrm{o},k}[\tau]}_{\text{intra-cell interference}}  \nonumber\\&+ \underbrace{\sum\nolimits_{l\in\mathcal{K},\,l \ne k} \bm{h}_{l,c_k}^H \left( \bm{\omega}_{l} s_{\mathrm{c},l}[\tau] + \bm{s}_{\mathrm{o},l}[\tau] \right)}_{\text{inter-cell interference}} 
		+ z_{c_k}[\tau],
	\end{align}
	where \( z_{c_k}[\tau] \sim \mathcal{CN}(0, \sigma_{\mathrm{c}}^2) \) denotes the additive white Gaussian noise, and $ \sigma_{\mathrm{c}}^2$ is the noise power at CU $c_k$. To capture the impact of practical receiver-side processing on system performance, we classify CUs into three types according to their interference cancellation capabilities with respect to pseudo-random dual-purpose signals generated a priori, as follows:
	\begin{itemize}
		\item \textbf{Type-I CUs} represent low-complexity and legacy receivers (e.g., IoT devices), which do not implement interference cancellation. Consequently, they are subject to both intra-cell and inter-cell interference.
		\item \textbf{Type-II CUs} correspond to more capable receivers equipped with successive interference cancellation (SIC) \cite{patel2002analysis}. These users can cancel intra-cell dual-purpose interference, i.e., $\bm{h}_{k,c_k}^H \bm{R}_{\mathrm{o},k} \bm{h}_{k,c_k}$, but still treat inter-cell dual-purpose interference, i.e., $\sum\nolimits_{l\in\mathcal{K},\,l \ne k} \bm{h}_{l,c_k}^H \bm{R}_{\mathrm{o},l} \bm{h}_{l,c_k}$ as noise.
		\item \textbf{Type-III CUs} denote advanced receivers capable of mitigating all dual-purpose interference, including both intra-cell and inter-cell components, i.e., $\sum\nolimits_{l\in\mathcal{K}} \bm{h}_{l,c_k}^H \bm{R}_{\mathrm{o},l} \bm{h}_{l,c_k}$ \cite{cheng2024optimal}.
	\end{itemize} 
	Based on this classification, the SINR expressions are given as $\mathrm{SINR}_{c_k}^{\mathrm{(I)}}(\{\bm{\omega}_{k},\bm{R}_{\mathrm{o},k}\})$ in~\eqref{sinr1}, $\mathrm{SINR}_{c_k}^{\mathrm{(II)}}(\{\bm{\omega}_{k},\bm{R}_{\mathrm{o},k}\})$ in~\eqref{sinr2}, and $\mathrm{SINR}_{c_k}^{\mathrm{(III)}}(\{\bm{\omega}_{k}\})$ in~\eqref{sinr3} for Type-I, Type-II, and Type-III CUs, respectively, at the top of next page.
    	\begin{figure*}[t]
        \begin{align}
            &\hspace{-10mm}\mathrm{SINR}_{c_k}^{\mathrm{(I)}}(\{\bm{\omega}_{k},\bm{R}_{\mathrm{o},k}\}) = \frac{\big|\bm{h}_{k,c_k}^H \bm{\omega}_{k}\big|^2}
            {\bm{h}_{k,c_k}^H \bm{R}_{\mathrm{o},k} \bm{h}_{k,c_k} + \sum\nolimits_{l\in\mathcal{K},\,l \ne k} \big|\bm{h}_{l,c_k}^H \bm{\omega}_{l}\big|^2 +\sum\nolimits_{l\in\mathcal{K},\,l \ne k} \bm{h}_{l,c_k}^H \bm{R}_{\mathrm{o},l} \bm{h}_{l,c_k} + \sigma_{\mathrm{c}}^2}\label{sinr1}.  \\
            &\hspace{-10mm}\mathrm{SINR}_{c_k}^{\mathrm{(II)}}(\{\bm{\omega}_{k},\bm{R}_{\mathrm{o},k}\}) = \frac{\big|\bm{h}_{k,c_k}^H \bm{\omega}_{k}\big|^2}
            { \sum\nolimits_{l\in\mathcal{K},\,l \ne k} \big|\bm{h}_{l,c_k}^H \bm{\omega}_{l}\big|^2 +\sum\nolimits_{l\in\mathcal{K},\,l \ne k} \bm{h}_{l,c_k}^H \bm{R}_{\mathrm{o},l} \bm{h}_{l,c_k} + \sigma_{\mathrm{c}}^2}\label{sinr2}. \\
            &\hspace{-10mm}\mathrm{SINR}_{c_k}^{\mathrm{(III)}}(\{\bm{\omega}_{k}\}) = \frac{\big|\bm{h}_{k,c_k}^H \bm{\omega}_{k}\big|^2}
            { \sum\nolimits_{l\in\mathcal{K},\,l \ne k} \big|\bm{h}_{l,c_k}^H \bm{\omega}_{l}\big|^2  + \sigma_{\mathrm{c}}^2}\label{sinr3}. 
        \end{align}
        \hrule
    \end{figure*}
	\subsection{Power Transfer}\label{wptsection}
	For ease of exposition, we adopt the typical linear energy-harvesting model \cite{ng2013wireless,clerckx2018fundamentals}, while noting that the proposed framework can be readily extended to nonlinear energy-harvesting characteristics. In particular, the harvested direct‑current (DC) power is assumed to be proportional to the total received RF signal power. More specifically, for ER \( e_k \) located at position \( \mathbf{p}_{e_k} \in \mathcal{A}_{e_k} \), the harvested power during the considered transmission block is given by
	\begin{align}
	\hspace{-3mm}	\mathcal{P}_{e_k}(\{\bm{\omega}_{l},\bm{R}_{\mathrm{o},l}\}) 
		&= \sum\nolimits_{l\in\mathcal{K}} \eta\, \bm{h}_{l,e_k}^H \left( \bm{\omega}_{l} \bm{\omega}_{l}^H + \bm{R}_{\mathrm{o},l} \right) \bm{h}_{l,e_k},
	\end{align}
	where \( \eta \in (0,1] \) denotes the RF-to-DC conversion efficiency.

	Furthermore, to account for location uncertainty of ERs in a realistic manner, we introduce the \emph{spatially averaged harvested power} as practically meaningful performance metric. In contrast to conventional worst‑case robust optimization \cite{Xiang6211384}, which guarantees performance  for the most adverse possible ER location and therefore often leads to overly conservative designs, the proposed metric here measures the average harvested energy across all possible ER positions within the uncertainty region. This captures the expected performance of a mobile ER whose exact location cannot be precisely determined, offering a realistic characterization of typical system behavior while avoiding the excessive conservatism of worst‑case robust optimization. For simplicity, we refer to this spatially averaged harvested power as the \emph{average harvested power} throughout the remainder of this paper, which is formally defined as
	\begin{equation}
		\bar{\mathcal{P}}_{e_k}(\{\bm{\omega}_{l},\bm{R}_{\mathrm{o},l}\})  = \mathbb{E}_{\mathbf{p}_{e_k} \in \mathcal{A}_{e_k}} \big\{ \mathcal{P}_{\mathbf{p}_{e_k}}(\{\bm{\omega}_{l},\bm{R}_{\mathrm{o},l}\}) \big\}.
	\end{equation}
	The analytical expression for \(\bar{\mathcal{P}}_{e_k}(\{\bm{\omega}_{l},\bm{R}_{\mathrm{o},l}\})  \) is provided in the following lemma.

	\begin{lemma}
		\label{average_power} The average harvested power by ER \(e_k\) is given by \begin{equation} \hspace{-3mm}\bar{\mathcal{P}}_{e_k}(\{\bm{\omega}_{l},\bm{R}_{\mathrm{o},l}\})= \eta\sum\nolimits_{l\in\mathcal{K}} \Tr\left(\bm{G}_{l,e_k}\big(\bm{\omega}_{l}\bm{\omega}_{l}^H+\bm{R}_{\mathrm{o},l}\big)\right), \end{equation} where the \((\alpha,\beta)\)-th entry of \(\bm{G}_{l,e_k}\in\mathbb{C}^{N\times N}\) is given by \eqref{channelcov}\footnote{It is worth noting that deriving a closed-form expression for \eqref{channelcov} is generally intractable, owing to the complexity of the underlying integral, especially for arbitrary regions \( \mathcal{A}_{e_k} \). Nevertheless, the integral can be efficiently evaluated using numerical tools, such as \textit{Mathematica}\cite{Mathematica}.} and $f_{e_k}(\mathbf{p})$=$1/{|\mathcal{A}_{e_k}|}$. 
	\end{lemma}
	\begin{figure*}[!t]
	\begin{equation}\label{channelcov}
        G_{l,e_k}[\alpha,\beta]
        = \int_{\mathcal{A}_{e_k}\subset\mathbb{R}^2}
        \frac{1}{4\kappa^2\|\mathbf{q}_{\alpha,l}-\mathbf{p}\|\,\|\mathbf{q}_{\beta,l}-\mathbf{p}\|}
        \exp\!\left(-j\kappa\!\left(\|\mathbf{q}_{\alpha,l}-\mathbf{p}\|-\|\mathbf{q}_{\beta,l}-\mathbf{p}\|\right)\right)
         f_{e_k}(\mathbf{p})\mathrm{d}\mathbf{p}.
    \end{equation}
		\hrule
	\end{figure*}
	\begin{proof}
		The proof follows directly by expressing the harvested power as $\mathbb{E}_{\mathbf{p}_{e_k}\in\mathcal{A}_{e_k}}\{\cdot\}$, applying the trace identity for scalar quadratic forms, and defining $\bm{G}_{l,e_k} = \mathbb{E}\{\bm{h}_{l,e_k}\bm{h}_{l,e_k}^H\}$ with entries given in~\eqref{channelcov}, at the top of next page.  
	\end{proof}
     \textcolor{black}{It is worth noting that the proposed framework is general and remains applicable to arbitrary ER spatial distributions and region geometries. For instance, if the ER positions follow a Gaussian distribution centered at $\bm{\mu}{e_k}$ with covariance matrix $\bm{\Sigma}_{e_k}$, then $f_{e_k}(\mathbf{p})$ takes the form of a standard 2D Gaussian density, given by $f_{e_k}(\mathbf{p})
     = \frac{1}{2\pi \sqrt{\det(\bm{\Sigma}_{e_k})}}
     \exp\!\left(-\tfrac{1}{2}(\mathbf{p}-\bm{\mu}_{e_k})^{\!T}\bm{\Sigma}_{e_k}^{-1}(\mathbf{p}-\bm{\mu}_{e_k})\right).
     $}
    
	\subsection{Target Detection}\label{ensensing}
	The considered ISCAP paradigm enables environmental sensing through a distributed MIMO radar architecture, in which spatially separated BSs cooperatively detect targets within the common sensing region defined in Section \ref{netlabel}. Compared with conventional co-located radar systems, the distributed configuration inherently offers enhanced spatial diversity and aligns naturally with multi-cell network deployments \cite{zhang9540344,meng2024cooperative}. However, fully exploiting distributed MIMO radar typically requires tight synchronization and high-capacity backhaul links to support the exchange of received raw data among BSs \cite{han2025network}. To address this challenge, we adopt a signal-level fusion strategy, where each BS locally extracts compact detection statistics from its received echoes and forwards them to a central unit for joint target detection \cite{cheng2024optimal}.

	Moreover, we consider a monostatic configuration at each BS, where both transmission and reception are performed leveraging the same antenna array. Hence, the echo signal corresponding to sampling point \(s_m\in\mathcal{S}\) and received by BS $k$ at symbol instance $\tau$ is given by
	\begin{align}\label{yeq}
		\bm{y}_{m,k}[\tau]
		=&\underbrace{\zeta \bm{h}_{k,s_m} \bm{h}_{k,s_m}^T \bm{s}_k[\tau]}_{\text{Direct-link echo signal}} + \underbrace{\sum\nolimits_{l\in\mathcal{K},\,l \ne k} \zeta \bm{h}_{k,s_m} \bm{h}_{l,s_m}^T \bm{s}_l[\tau]}_{\text{Cross-link echo signals}} \nonumber\\
		&+ \bm{z}_{\mathrm{s},k}[\tau],
	\end{align}
	where $\zeta$ denotes the radar cross-section (RCS) and is assumed to be identical across all sampling points in $\mathcal{A}_{\rm s}$;  \( \bm{z}_{\mathrm{s},k}[\tau] \sim \mathcal{CN}(\bm{0}, \sigma_{\mathrm{s}}^2 \bm{I}) \) denotes the additive white Gaussian noise with power $\sigma_{\mathrm{s}}^2 $. The first term on the right-hand side of \eqref{yeq} corresponds to the direct reflection from the target back to its originating BS, and the second term represents the cross-link echo signals transmitted by other BSs, reflected by the target, and subsequently received by BS $k$. In the considered ISCAP design, both the information-bearing communication signal and the dual‑purpose signal contribute to the direct-link echoes; hence, reflections of communication transmissions from the target are also exploited to enhance the sensing power. Note that the direct BS-to-BS leakage is assumed to be fully known and can be digitally removed prior to sensing, and thus it is excluded from the received-signal model \cite{cheng2024optimal}. Once the echo signals are received, each BS adopts a matched filter to maximize the signal-to-noise ratio (SNR) by correlating the incoming signals with a replica of the transmitted waveform. Accordingly, each BS processes only the direct-link echoes (i.e., BS-to-target-to-originating BS) through matched filtering, while discarding  cross-link reflections from other BSs \cite{cheng2024optimal}. The resulting match-filtered outputs are then forwarded to a central unit, which performs joint detection by fusing the aggregated observations from all BSs across the sampled spatial locations.
	
	To evaluate the sensing performance at each sampled spatial location, we adopt an energy-based detection criterion \cite{liu2022integrated}. Specifically, the presence of a target is determined by comparing the received echo power against a predefined threshold, which is determined by the desired false alarm probability, denoted by \( P_{\mathrm{FA}} \). Under this criterion, the detection probability corresponding to point \( s_m \), denoted by \( P_{{\rm D},m} \), is given by \cite{cheng2024optimal}
	\begin{equation}\label{detectionprob}
		{P}_{{\rm D},m}={Q}\left({Q}^{-1}\big({P}_{\rm FA}\big)-\sqrt{\frac{2 {\varphi}_{m}}{\sigma_{\mathrm{s}}^2}}\right),
	\end{equation}
	where ${\varphi}_{m}$ is the total received echo power contributed by the direct links at all coordinated BSs, as given in \eqref{echopower} at the top of the next page, and ${Q}(\cdot)$ denotes the ${Q}$-function, i.e., $
	{Q}(x) = \frac{1}{\sqrt{2\pi}} \int_x^{\infty} \exp\left( -\frac{t^2}{2} \right) \mathrm{d}t$.
	To characterize the sensing performance over the designated region, we define the worst-case detection probability across all sampled spatial locations as $
		P_{\mathrm{D}} = \min_{s_m \in \mathcal{S}} P_{{\rm D},m}$, 
	which represents the lowest probability of successful target detection and thus serves as a robust metric for network-level sensing evaluation \cite{cheng2024optimal,liu2022integrated}.
	
	It is observed from \eqref{detectionprob} that the detection probability \( P_{{\rm D},m} \) increases monotonically with the aggregated received echo power \( \varphi_m(\{\bm{\omega}_{k},\bm{R}_{\mathrm{o},k}\}) \). Consequently, maximizing the worst-case detection probability is equivalent to maximizing the minimum received echo power across all sampled points, i.e., \( \min_{s_m \in \mathcal{S}} \varphi_m(\{\bm{\omega}_{k},\bm{R}_{\mathrm{o},k}\}) \). This quantity is thus adopted as the sensing design metric, while the resulting worst-case detection probability \( P_{\mathrm{D}} \) is illustrated in the numerical results.
	\begin{figure*}[t]
		\begin{equation}\label{echopower}
			\varphi_{m}(\{\bm{\omega}_{k},\bm{R}_{\mathrm{o},k}\})
			=  |\zeta|^2\sum\nolimits_{k=1}^{K} \operatorname{Tr} \big(\bm{h}_{k,s_m}^* \bm{h}_{k,s_m}^T (\bm{\omega}_{k} \bm{\omega}_{k}^H + \bm{R}_{\mathrm{o},k} )  \big)\sum\nolimits_{n=1}^{N} \frac{\lambda^2}{16 \pi^2  \|\mathbf{q}_{n,k} - \mathbf{p}_{s_m}\|^2}.
		\end{equation}
		\hrule
	\end{figure*}
	\vspace{-6mm}
	\subsection{Problem Formulation}
	We now introduce a unified optimization framework, where the key design variables are the information beamforming vectors, i.e., \( \{ \bm{\omega}_{k} \} \), and the covariance matrices of the dual-purpose signals, i.e., \( \{ \bm{R}_{\mathrm{o},k} \} \). More specifically, the  objective is to maximize the worst-case sensing performance, quantified by the minimum aggregate echo signal power across all sampled points, subject to the following practical constraints: (i) the SINR at each CU must exceed a prescribed threshold; (ii) each ER must harvest at least the required average power under location uncertainty; (iii) the transmit power of each BS remains within its assigned budget. Based on the SINR models introduced in Section~\ref{cutype}, we formulate three separate optimization problems corresponding to Type-I, Type-II and Type-III CUs, respectively, as detailed below:
	\begin{subequations} \label{P1}
		\begin{align}
			\hspace{-5mm}(\text{P1}): \quad
			\max_{\{ \bm{\omega}_{k}, \bm{R}_{\mathrm{o},k} \}} \quad & \min_{s_m\in\mathcal{S}}\quad \varphi_{m}(\{ \bm{\omega}_{k}, \bm{R}_{\mathrm{o},k} \}) \\
			\text{s.t.}\hspace{5mm} \quad 
			&	\hspace{-5mm}\mathrm{SINR}^{\mathrm{(I)}}_{c_k}(\{ \bm{\omega}_{k}, \bm{R}_{\mathrm{o},k} \}) \ge \Gamma_{c_k},  \forall c_k \in \mathcal{C}, \label{P1a}  \\[2pt]
			& \bar{\mathcal{P}}_{e_k}(\{ \bm{\omega}_{k}, \bm{R}_{\mathrm{o},k} \}) \geq \Omega_{e_k},  \forall e_k \in \mathcal{E}, \label{P1b} \\[2pt]
			& \|\bm{\omega}_{k}\|^2 + \mathrm{Tr}(\bm{R}_{\mathrm{o},k}) \le P_{\max}, \forall k \in \mathcal{K}, \label{P1c} \\[2pt]
			& \bm{R}_{\mathrm{o},k} \succeq \bm{0}, \quad \forall k \in \mathcal{K}, \label{P1d}
		\end{align}
	\end{subequations}
	\vspace{-5mm}
	\begin{subequations} \label{P2}
		\begin{align}
			\hspace{-4mm} (\text{P2}): \quad
			\max_{\{ \mathrm{\omega}_{k}, \bm{R}_{\mathrm{o},k} \}} \quad & \min_{s_m\in\mathcal{S}}\quad \varphi_{m}(\{ \bm{\omega}_{k}, \bm{R}_{\mathrm{o},k} \}) \\[2pt]
			\text{s.t.}\hspace{5mm} \quad 
			&\hspace{-5mm} \mathrm{SINR}^{\mathrm{(II)}}_{c_k}(\{ \bm{\omega}_{k}, \bm{R}_{\mathrm{o},k} \}) \ge \Gamma_{c_k},  \forall c_k \in \mathcal{C}, \label{P2a} \\[2pt]
			& \eqref{P1b} - \eqref{P1d}, \nonumber
		\end{align}
	\end{subequations}
	\vspace{-5mm}
	\begin{subequations} \label{P3}
		\begin{align}
			\hspace{-4mm}  (\text{P3}): \quad
			\max_{\{ \mathrm{\omega}_{k}, \bm{R}_{\mathrm{o},k} \}} \quad & \min_{s_m\in\mathcal{S}}\quad \varphi_{m}(\{ \bm{\omega}_{k}, \bm{R}_{\mathrm{o},k} \}) \\[2pt]
			\text{s.t.} \quad 
			& \mathrm{SINR}^{\mathrm{(III)}}_{c_k}(\{ \bm{\omega}_{k} \}) \ge \Gamma_{c_k},  \forall c_k \in \mathcal{C},  \\[2pt]
			& \eqref{P1b} - \eqref{P1d}, \nonumber
		\end{align}
	\end{subequations}
	where \( \Gamma_{c_k} \) and \( \Omega_{e_k} \) represent the SINR and energy harvesting thresholds for CU \( c_k \) and ER \( e_k \), respectively. 
	
	It is important to emphasize that the above optimization problems are inherently non-convex and challenging to solve in their original forms. This complexity arises primarily from the SINR constraints, which involve fractional quadratic expressions in the beamforming vectors. Consequently, conventional convex optimization techniques are not directly applicable, necessitating the development of specialized solution methods.
	\section{Optimal and Suboptimal Coordinated Beamforming Designs}\label{optimalsolu}
	In this section, we propose two effective solution approaches to address the formulated optimization problems $(\mathrm{P1})-(\mathrm{P3})$. Specifically, we first develop an optimal beamforming strategy leveraging the SDR technique, which transforms the original problems into tractable convex programs. To further enhance the scalability in large-scale deployments, we then introduce a low-complexity MRT-based sub-optimal scheme. Finally, we provide a comparative computational complexity of both methods to offer practical insights for system implementation.
	\vspace{-9mm}
	\subsection{Optimal SDR-Based Beamforming Design}
	In the following, we present the optimal beamforming design capitalizing on the SDR technique. To facilitate reformulation, we introduce an auxiliary variable \(\Theta\), and define \( \bm{\mathcal{W}}_k = \bm{\omega}_k \bm{\omega}_k^H \succeq \bm{0} \), $\forall$\( k \in \mathcal{K} \), as well as \( \bm{\mathcal{H}}_{l,s_m} = \bm{h}_{l,s_m}^* \bm{h}_{l,s_m}^T \). Problems $(\mathrm{P1})$, $(\mathrm{P2})$, and $(\mathrm{P3})$ are then equivalently reformulated as follows:
	\begin{subequations}
		\begin{align}
			\hspace{-2mm}(\text{P1.1}): \quad\hspace{-5mm}
			\max_{\{ \bm{\mathcal{W}}_k, \bm{R}_{\mathrm{o},k}\}, \Theta}& \quad  \Theta  \\
			\text{s.t.}\quad \quad\quad\quad& \hspace{-12mm}
			\sum\nolimits_{l\in\mathcal{K}}\sum\nolimits_{n=1}^{N}\frac{\lambda^2|\zeta|^2}{16 \pi^2  \|\mathbf{q}_{n,l} - \mathbf{p}_{s_m}\|^2} \nonumber\\
			&\hspace{-12mm}\Tr \left( \bm{\mathcal{H}}_{l,s_m} (\bm{\mathcal{W}}_l + \bm{R}_{\mathrm{o},l} )  \right) \ge\ {} \Theta,  \forall s_m\in\mathcal{S}, \label{p11b} \\
			&\hspace{-12mm}\sum\nolimits_{l\in\mathcal{K}} \mathrm{Tr}(\bm{h}_{l,c_k} \bm{h}_{l,c_k}^H \bm{\mathcal{W}}_l) \nonumber\\
			&\hspace{-12mm}+  \sum\nolimits_{l\in\mathcal{K}}\mathrm{Tr}(\bm{h}_{l,c_k} \bm{h}_{l,c_k}^H \bm{R}_{\mathrm{o},l})+ \sigma_{\mathrm{c}}^2 \nonumber\\
			&\hspace{-12mm}\le\frac{1+\Gamma_{c_k}}{\Gamma_{c_k}}\mathrm{Tr}(\bm{h}_{k,c_k} \bm{h}_{k,c_k}^H \bm{\mathcal{W}}_k),\forall c_k \in\mathcal{C},  \\
			&\hspace{-13mm}\eta\sum\nolimits_{l\in\mathcal{K}} \Tr\left(\bm{G}_{l,e_k}\big(\bm{\mathcal{W}}_{l}+\bm{R}_{\mathrm{o},l}\big)\!\right)\!\!\ge\Omega_{e_k},\forall e_k\in\mathcal{E},\label{p11d}\\
			&\hspace{-12mm}\mathrm{Tr}(\bm{\mathcal{W}}_{k}+\bm{R}_{\mathrm{o},k}) \le P_{\max},\forall k\in\mathcal{K},\label{p11e}  \\
			&\hspace{-12mm}\bm{R}_{\mathrm{o},k} \succeq \bm{0}, \hspace{6mm} \forall k\in\mathcal{K},\label{p11f}\\
			&\hspace{-12mm}\bm{\mathcal{W}}_k \succeq \bm{0},  \hspace{7mm} \forall k\in\mathcal{K},\label{p11g}\\
			&\hspace{-12mm}\rank(\bm{\mathcal{W}}_k)\le1, \quad\forall k\in\mathcal{K},\label{p11h}
		\end{align}
	\end{subequations}
	\vspace{-5mm}
	\begin{subequations}
		\begin{align}
			\hspace{-3mm} (\text{P2.1}): \quad\hspace{-3mm}
			\max_{\{ \bm{\mathcal{W}}_k, \bm{R}_{\mathrm{o},k}\}, \Theta}& \quad  \Theta  \\
			\text{s.t.} \quad\quad\quad\quad&\hspace{-12mm}\sum\nolimits_{l\in\mathcal{K} } \mathrm{Tr}(\bm{h}_{l,c_k} \bm{h}_{l,c_k}^H \bm{\mathcal{W}}_l) \nonumber\\
			&\hspace{-12mm}+ \sum\nolimits_{l \in\mathcal{K},l\ne k }\mathrm{Tr}(\bm{h}_{l,c_k} \bm{h}_{l,c_k}^H \bm{R}_{\mathrm{o},l})+ \sigma_{\mathrm{c}}^2 \nonumber\\
			&\hspace{-12mm}\le \frac{1+\Gamma_{c_k}}{\Gamma_{c_k}}\mathrm{Tr}(\bm{h}_{k,c_k} \bm{h}_{k,c_k}^H \bm{\mathcal{W}}_k),\forall c_k \in\mathcal{C},  \\
			&\hspace{-12mm}\eqref{p11b},\ {} \eqref{p11d}-\eqref{p11h}\nonumber,
		\end{align}
	\end{subequations}
	\vspace{-5mm}
	\begin{subequations}
		\begin{align}
			\hspace{-3mm}(\text{P3.1}): \quad\hspace{-3mm}
			\max_{\{ \bm{\mathcal{W}}_k, \bm{R}_{\mathrm{o},k}\}, \Theta}& \quad  \Theta  \\
			\text{s.t.} \quad\quad\quad\quad&\hspace{-12mm}\sum\nolimits_{l\in\mathcal{K} } \mathrm{Tr}(\bm{h}_{l,c_k} \bm{h}_{l,c_k}^H \bm{\mathcal{W}}_l) + \sigma_{\mathrm{c}}^2 \nonumber\\
			&\hspace{-12mm}\le \frac{1+\Gamma_{c_k}}{\Gamma_{c_k}}\mathrm{Tr}(\bm{h}_{k,c_k} \bm{h}_{k,c_k}^H \bm{\mathcal{W}}_k), \forall c_k \in\mathcal{C},  \\
			&\hspace{-12mm}\eqref{p11b},\ {} \eqref{p11d}-\eqref{p11h}\nonumber,
		\end{align}
	\end{subequations}
	respectively. 	Note that problems $(\mathrm{P1.1})$, $(\mathrm{P2.1})$, and $(\mathrm{P3.1})$ remain non-convex due to the rank-one constraints in \eqref{p11h}. To obtain tractable formulations, we first relax these constraints, leading to the SDR versions, denoted as \((\mathrm{SDR1.1})\), \((\mathrm{SDR2.1})\), and \((\mathrm{SDR3.1})\), respectively, such all three problems are convex that can be efficiently solved using standard convex optimization tools \cite{cvx}. Let \( \{ \{\bm{\mathcal{W}}_k^{\mathrm{I}}\}, \{\bm{R}_{\mathrm{o},k}^{\mathrm{I}}\}, \Theta^{\mathrm{I}} \} \), \( \{ \{\bm{\mathcal{W}}_k^{\mathrm{II}}\}, \{\bm{R}_{\mathrm{o},k}^{\mathrm{II}}\}, \Theta^{\mathrm{II}} \} \), and \( \{ \{\bm{\mathcal{W}}_k^{\mathrm{III}}\}, \{\bm{R}_{\mathrm{o},k}^{\mathrm{III}}\}, \Theta^{\mathrm{III}} \} \) denote the optimal solutions to \((\mathrm{SDR1.1})\), \((\mathrm{SDR2.1})\), and \((\mathrm{SDR3.1})\), respectively.

	It is important to emphasize that the matrices \( \{\bm{\mathcal{W}}_k^{\mathrm{I}}\} \), \( \{\bm{\mathcal{W}}_k^{\mathrm{II}}\} \), and \(\{ \bm{\mathcal{W}}_k^{\mathrm{III}}\} \) are not necessarily rank-one. As such, they generally do not satisfy the original rank-one constraints in $(\mathrm{P1.1})-(\mathrm{P3.1})$, and cannot be directly exploited as feasible solutions to the corresponding original problems. To recover solutions that satisfy the original rank-one constraints, we develop post-processing procedures that construct equivalent rank-one beamforming matrices from the SDR outputs while preserving the optimal objective values. The details are provided in the following proposition.
	\begin{proposition}\label{prop}
		The semidefinite relaxation \((\mathrm{SDR1.1})\) of problem \((\mathrm{P1.1})\) is tight. Specifically, if any $\bm{\mathcal{W}}_k^{\mathrm{I}}$ for $k \in \mathcal{K}$ is not rank-one, we can construct an alternative solution $\{\{\widetilde{\bm{\mathcal{W}}}_k\}, \{\widetilde{\bm{R}}_{\mathrm{o},k}\}, \widetilde{\Theta}\}$, which satisfies all constraints of \((\mathrm{P1.1})\) and achieves the same objective value.
		\begin{subequations}
			\begin{align}
				&\widetilde{\bm{\omega}}_k = { \bm{\mathcal{W}}_k^\dagger \bm{h}_{k,c_k} }\left(\bm{h}_{k,c_k}^H \bm{\mathcal{W}}_k^\dagger \bm{h}_{k,c_k}\right)^{-\frac{1}{2}}, \\
				&\widetilde{\bm{\mathcal{W}}}_k = \widetilde{\bm{\omega}}_k \widetilde{\bm{\omega}}_k^H, \\
				&\widetilde{\bm{R}}_{\mathrm{o},k} = \bm{\mathcal{W}}_k^\dagger + \bm{R}_{\mathrm{o},k}^\dagger - \widetilde{\bm{\mathcal{W}}}_k, \label{eq:rank-adjustment} \\
				&\widetilde{\Theta} = \Theta^\dagger.
			\end{align}
		\end{subequations}
	\end{proposition}
	\begin{proof}
		Please refer to \ref{proofofprop}.	
	\end{proof}
	The conclusion established in Proposition \ref{prop} also applies to \((\mathrm{P2.1})\) and \((\mathrm{P3.1})\), for which the details are omitted for brevity. Hence, solving $(\mathrm{SDR1.1})$, $(\mathrm{SDR2.1})$, and $(\mathrm{SDR3.1})$ directly yields the globally optimal value of their corresponding original problems $(\mathrm{P1.1})$, $(\mathrm{P2.1})$, and $(\mathrm{P3.1})$, respectively. 
	
	The only distinction among $(\mathrm{SDR1.1})$, $(\mathrm{SDR2.1})$, and $(\mathrm{SDR3.1})$ lies in the SINR constraints, which capture the CUs' capability to mitigate interference. Since SIC eliminates non-negative interference terms in the SINR models, the SINR constraint in $(\mathrm{SDR2.1})$ is a relaxation of that in $(\mathrm{SDR1.1})$, and the constraint in $(\mathrm{SDR3.1})$ is in turn a relaxation of that in $(\mathrm{SDR2.1})$. Consequently, every feasible point of $(\mathrm{SDR1.1})$ is also feasible for $(\mathrm{SDR2.1})$ and $(\mathrm{SDR3.1})$. As all three problems maximize the same objective $\Theta$, their optimal values satisfy $\Theta^{\mathrm{III}} \geq \Theta^{\mathrm{II}} \geq \Theta^{\mathrm{I}}$. Interestingly, despite the relaxation, $(\mathrm{SDR1.1})$ and $(\mathrm{SDR2.1})$ achieve the same optimal objective value, as summarized in the following corollary.

	\begin{corollary}\label{corollary}
		The optimal objective values of \((\mathrm{SDR1.1})\) and \((\mathrm{SDR2.1})\) are identical, i.e., $\Theta^{\mathrm{I}} = \Theta^{\mathrm{II}}$.
	\end{corollary}
	\begin{proof}
		See Appendix \ref{corollaryproof}.
	\end{proof}
	The equality \(\Theta^{\mathrm{I}} = \Theta^{\mathrm{II}}\) indicates that, under the optimal SDR solution, the dual-purpose transmit covariance \(\bm{R}_{o,k}\) at each BS lies entirely in the null space spanned by the associated CU’s channel, i.e., \(\bm{h}_{k,c_k}^{H}\bm{R}_{o,k}\bm{h}_{k,c_k}=0\). As a result, intra-cell interference is already eliminated through transmit-side beamforming optimization, and the additional cancellation capability available at Type-II CUs offers no further improvement over Type-I CUs.
	\subsection{Suboptimal MRT-Based Beamforming Design}  \label{suboptimalsolu}
	While the previous subsection established the optimal solution via the SDR technique, here we propose a practical and low-complexity alternative based on MRT \cite{Lo1999MRT}. In this approach, the information beamforming vectors are fixed to align with the channels of their respective CUs according to the following MRT principle. The dual-purpose signal is constructed by exploiting two predefined beams, with one dedicated to wireless power transfer, aligned with the dominant eigenvector of location-averaged channel matrix of ERs; and the other dedicated to sensing, oriented along the principal eigenvector of the aggregated round-trip response.

	More specifically, the information signal at BS $k$ is transmitted leveraging a fixed beamforming vector \( \bm{\omega}_{k} = \sqrt{\varrho_{c,k}} \frac{\bm{h}_{k,c_k}}{\|\bm{h}_{k,c_k}\|} \), where \( \varrho_{c,k} \ge 0 \) denotes the allocated power for information transfer that will be optimized. The covariance matrix of dual-purpose signal transmitted by BS~\(k\) is structured as 
	\begin{equation}\label{requ}
		\bm{R}_{\mathrm{o},k} = \varrho_{e,k} \bm{\nu}_{e,k} \bm{\nu}_{e,k}^H + \varrho_{s,k} \bm{\nu}_{s,k} \bm{\nu}_{s,k}^H,
	\end{equation}
	where \( \bm{\nu}_{e,k} \in\mathbb{C}^{N\times 1}\) and \( \bm{\nu}_{s,k} \in\mathbb{C}^{N\times 1}\) are unit-norm vectors defining the fixed transmission directions for WPT and sensing, respectively. Here, $\bm{\nu}_{e,k}$ is selected as the principal eigenvector of \( \bm{G}_{k,e_k} \), and $\bm{\nu}_{s,k}$ is derived as the dominant eigenvector of the spatially averaged round-trip matrix, i.e., \( \bm{A}_k = \frac{1}{M} \sum_{s_m\in\mathcal{S}} \frac{|\zeta|^2\lambda^2}{16 \pi^2  } \sum\nolimits_{n=1}^{N} \frac{\bm{\mathcal{H}}_{k,s_m}}{\|\mathbf{q}_{n,k} - \mathbf{p}_{s_m}\|^2} \), which captures the dominant echo direction across all sensing points. \textcolor{black}{Intuitively, selecting $\bm{\nu}_{e,k}$ as the principal eigenvector of $\bm{G}_{k,e_k}$ aligns the WPT beam with the statistically strongest ER channel direction and maximizes the spatially averaged harvested power at ER $e_k$. Similarly, choosing $\bm{\nu}_{s,k}$ as the dominant eigenvector of $\bm{A}_k$ steers the sensing beam toward the most significant echo direction over the sampled sensing region, thereby enhancing the average received echo power.}  Furthermore, the scalars \( \varrho_{e,k} \ge 0 \) and \( \varrho_{s,k} \ge 0 \) in \eqref{requ} represent the power allocations to these two beams, which will also be optimized. Notice that the design problem reduces to optimizing the scalar power allocation among communication, energy transfer, and sensing at each BS. To this end, we formulate three power allocation problems, corresponding to the cases with Type-I, Type-II, and Type-III CUs, respectively, in the following,
	\begin{subequations} \label{P-MRT-Explicit}
		\allowdisplaybreaks
		\begin{align}
			\hspace{-3mm}   (\mathrm{P4}): \quad
			\hspace{-5mm}\max_{\{ \varrho_{c,k}, \varrho_{e,k}, \varrho_{s,k}\} ,\Theta} \hspace{-3mm} &\quad \quad \Theta \\
			\text{s.t.} \quad\quad\quad &\hspace{-5mm}\sum\nolimits_{n=1}^{N}\sum\nolimits_{l\in\mathcal{K},l\neq k}
			\frac{\lambda^2|\zeta|^2}{16 \pi^2  \|\mathbf{q}_{n,l} - \mathbf{p}_{s_m}\|^2}\nonumber\\
			&\hspace{-5mm}\bigg(\hspace{-1mm}
			\varrho_{c,l} \frac{\bm{h}_{l,c_l}^H \bm{\mathcal{H}}_{l,s_m} \bm{h}_{l,c_l}}{\|\bm{h}_{l,c_l}\|^2}+ \varrho_{e,l}  \bm{\nu}_{e,l}^H \bm{\mathcal{H}}_{l,s_m} \bm{\nu}_{e,l}\nonumber\\
			&\hspace{-5mm}+ \varrho_{s,l}  \bm{\nu}_{s,l}^H \bm{\mathcal{H}}_{l,s_m} \bm{\nu}_{s,l}
			\bigg)
			\ge \Theta,\forall s_m \in \mathcal{S},\label{24b}\\
			& \hspace{-5mm}\sum\nolimits_{l \in \mathcal{K},l\neq k} 
			\varrho_{c,l} \frac{|\bm{h}_{l,c_k}^H \bm{h}_{l,c_l}|^2}{\|\bm{h}_{l,c_l}\|^2} 
			+ \sigma_{\mathrm{c}}^2\nonumber\\
			&\hspace{-5mm}+ \sum\nolimits_{l \in \mathcal{K}}\varrho_{e,l}|\bm{h}_{l,c_k}^H \bm{\nu}_{e,l}|^2\hspace{-1mm}+\hspace{-1mm}\varrho_{s,l} |\bm{h}_{l,c_k}^H \bm{\nu}_{s,l}|^2 \big)\nonumber\\
			&\hspace{-5mm}\le\frac{\varrho_{c_k} \|\bm{h}_{k,c_k}\|^2}{\Gamma_{c_k}},\forall c_k \in\mathcal{C},  \\
			& \hspace{-5mm}\sum\nolimits_{l \in \mathcal{K}}  \bigg( 
			\varrho_{c,l} \frac{ \bm{h}_{l,c_l}^H \bm{G}_{l,e_k} \bm{h}_{l,c_l} }{ \|\bm{h}_{l,c_l}\|^2 }\nonumber\\
			&\hspace{-5mm} +\varrho_{e,l} \bm{\nu}_{e,l}^H \bm{G}_{l,e_k} \bm{\nu}_{e,l}+ \varrho_{s,l} \bm{\nu}_{s,l}^H \bm{G}_{l,e_k} \bm{\nu}_{s,l}\bigg)\nonumber\\
			&\hspace{-5mm}\ge \frac{\Omega_k}{\eta},
			\forall e_k\in\mathcal{E},\label{24d}\\
			&\hspace{-5mm} \varrho_{c,k}\! +\! \varrho_{e,k}\! +\! \varrho_{s,k} \!\le\! P_{\max},  \forall k\in\mathcal{K}, \label{24e} \\
			&\hspace{-5mm} \varrho_{c,k},\, \varrho_{e,k},\, \varrho_{s,k} \ge 0,\quad \forall k\in\mathcal{K}, \label{24f}
		\end{align}
	\end{subequations}
	\begin{subequations} 
		\begin{align}
			(\mathrm{P5}): \quad\quad\quad
			\hspace{-5mm}\max_{\{ \varrho_{c,k}, \varrho_{e,k}, \varrho_{s,k}\},\Theta}  & \quad\quad \Theta \\
			\text{s.t.}\quad\quad\quad \quad 
			&\hspace{-12mm} \sum\limits_{l \in \mathcal{K},l\ne k} 
			\varrho_{c,l} \frac{|\bm{h}_{l,c_k}^H \bm{h}_{l,c_l}|^2}{\|\bm{h}_{l,c_l}\|^2} 
			+ \sigma_{\mathrm{c}}^2+\nonumber\\
			& \hspace{-12mm}\sum\limits_{l \in \mathcal{K},l\ne k}\hspace{-3mm}\big(\varrho_{e,l}|\bm{h}_{l,c_k}^H \bm{\nu}_{e,l}|^2+\varrho_{s,l} |\bm{h}_{l,c_k}^H \bm{\nu}_{s,l}|^2\big)\nonumber\\
			&\hspace{-12mm}\le
			\frac{\varrho_{c,k} \|\bm{h}_{k,c_k}\|^2}{\Gamma_{c_k}},	\forall c_k \in\mathcal{C},  \\
			&\hspace{-12mm}\eqref{24b}, \,\eqref{24d}-\eqref{24f},\nonumber
		\end{align}
	\end{subequations} 
	\begin{subequations} 
		\begin{align}
			(\mathrm{P6}): \quad
			\hspace{3mm}\max_{\{ \varrho_{c,k}, \varrho_{e,k}, \varrho_{s,k}\},\Theta}  & \quad\quad \Theta \\
			\text{s.t.}\hspace{5mm} \quad 
			&\hspace{-5mm} \sum\limits_{l \in \mathcal{K},l\ne k} 
			\varrho_{c,l} \frac{|\bm{h}_{l,c_k}^H \bm{h}_{l,c_l}|^2}{\|\bm{h}_{l,c_l}\|^2} 
			+ \sigma_{\mathrm{c}}^2\nonumber\\
			&
			\hspace{-5mm}\le
			\frac{\varrho_{c,k} \|\bm{h}_{k,c_k}\|^2}{\Gamma_{c_k}},	\forall c_k \in\mathcal{C},  \\
			&\hspace{-5mm}\eqref{24b}, \,\eqref{24d}-\eqref{24f}.\nonumber
		\end{align}
	\end{subequations}
	Notice that problems~$\mathrm{(P4)}$, $\mathrm{(P5)}$ and $\mathrm{(P6)}$ are linear programs (LPs) in the scalar variables \( \{\varrho_{c,k}, \varrho_{e,k}, \varrho_{s,k}\} \), along with the auxiliary optimization variable \( \Theta \).  Owing to the fixed beam directions, all quadratic terms reduce to scalar constants, transforming the constraints into affine functions of the power-allocation variables. As a result, all three problems can be efficiently solved using standard LP numerical solvers to directly obtain the power allocation for the suboptimal MRT-based design.
	
	It is worth emphasizing that  when the CUs, ERs, and sensing points are located at distinct positions, the near-field channel vectors become asymptotically orthogonal as the number of antennas $N \to \infty$ \cite{liu2023near,wu2023multiple}. Specifically,  $\forall\,l \ne k$, the normalized inner product $|\bm{h}_{l, c_k}^H \bm{h}_{l, c_l}|^2 / \|\bm{h}_{l, c_l}\|^2$ and the squared inner products $|\bm{h}_{l, c_k}^H \bm{\nu}_{e,l}|^2$ and $|\bm{h}_{l, c_k}^H \bm{\nu}_{s,l}|^2$ all vanish in the limit $N \to \infty$. As a result, the interference terms in the SINR constraint become negligible in the large-array regime, rendering problems~$\mathrm{(P4)}$--$\mathrm{(P5)}$ asymptotically equivalent. Furthermore, a similar orthogonality effect arises at the ERs and sensing points, where signals not intended for a given ER or sensing location, such as information beams or transmissions from other BSs, contribute negligibly to its received power. Therefore, the powering and sensing functionalities depend solely on their respective designated beams, and the associated constraints become effectively decoupled among the BSs.   Such asymptotic decoupling enables  closed-form power allocation solutions at each BS based solely on local channel parameters, as described in the following remarks.
	
	\begin{remark}\label{re1}
		For the MRT-based approach, the asymptotic orthogonality of near-field channels leads to simple closed-form power allocation at BS~\(k\), which is given by $\varrho_{c,k}^*=\frac{\Gamma_{c_k}\sigma_{\mathrm{c}}^{2}}{\|\bm h_{k,c_k}\|^{2}}$, $
		\varrho_{e,k}^*=\frac{\Omega_k/\eta}{\bm\nu_{e,k}^{H}\bm G_{k,e_k}\bm\nu_{e,k}}$, and $\varrho_{s,k}^*=P_{\max}-\varrho_{c,k}^*-\varrho_{e,k}^*
		$. 
	\end{remark}
	
	\begin{remark}
		The closed-form solution in Remark \ref{re1} is feasible if and only if $
		\varrho_{c_k}^* \le P_{\max}$, $\varrho_{e_k}^* \le P_{\max}$ and $\varrho_{e_k}^*+\varrho_{e_k}^* \le P_{\max}$, $\forall c_k\in\mathcal{C}$, $e_k\in\mathcal{E}$, respectively.
	\end{remark}

	\subsection{Complexity Analysis}
	We now provide a detailed computational complexity analysis of the proposed solutions and discuss their implications for practical implementation. For the SDR-based design, which involves solely an SDP, the per-iteration computational complexity is on the order of $\mathcal{O}((KN)^6)$, by using standard interior-point methods \cite{boyd2004convex}. In contrast, for the MRT-based approach, the resulting LP has a complexity of $\mathcal{O}(K^3)$, and in the large-array regime, where a closed-form solution exists, the complexity is further reduced to $\mathcal{O}(K)$.
	
	\textcolor{black}{Hence, while the SDR-based approach achieves optimal performance through precise convex optimization, the MRT-based solution provides enhanced scalability and offers a favorable trade-off between computational complexity and performance, making it attractive for large-array or dense network deployments. Nevertheless, both schemes have polynomial-time computational complexity and can be efficiently executed for most practical system dimensions envisioned in future 6G networks.}

	\begin{table}[t]
		\centering
		\caption{System Parameters}
		\setlength{\extrarowheight}{2pt}
		\label{tabparameters}
		\begin{tabular}{|l|l|l|}
			\hline
			\textbf{Parameter} & \textbf{Symbol} & \textbf{Value} \\
			\hline
			Number of antenna elements & $N$ & 64 \\
			\hline
			Number of BSs & $K$ & 3 \\
			\hline
			Space between antenna elements & $d$ & $0.0625$ m \\
			\hline
			RF-to-DC conversion efficiency & $\eta$ & 0.7 \\
			\hline
			RCS  & $|\zeta|$ & 1 \\
			\hline
			Transmit power budget & $P_{\max}$ & 27 dBm \\
			\hline
			CU noise power & $\sigma_{\mathrm{c}}^2$ & -50 dBm \\
			\hline
			Sensing noise power & $\sigma_{\mathrm{s}}^2$ & -97 dBm \\
			\hline
			Carrier frequency & $f$ & 2.4 GHz \\
			\hline
			Rayleigh distance & - & $248.1$\,m\\
			\hline
		\end{tabular}
		\vspace{-8pt}   
	\end{table}

	\begin{figure*}[!t]
		\centering
		\subfloat[Case~1.]{%
			\includegraphics[width=0.3\linewidth]{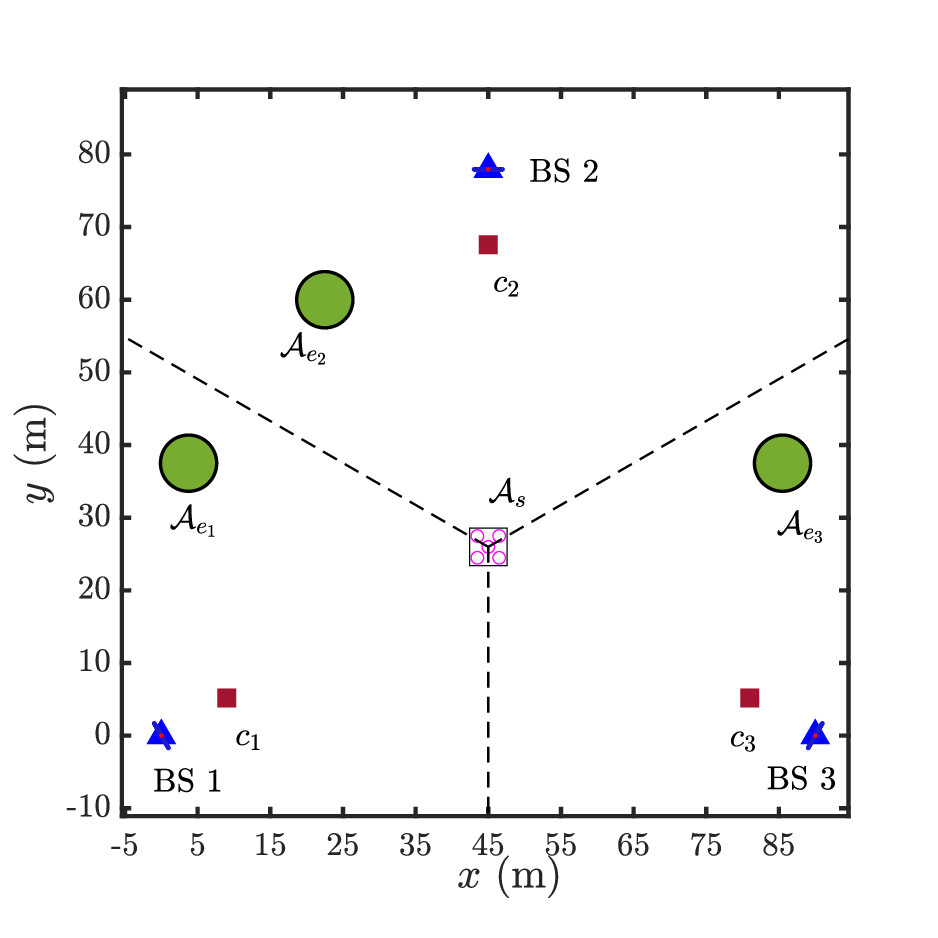}%
			\label{layout_case1}}
		\hfil
		\subfloat[Case~2.]{%
			\includegraphics[width=0.3\linewidth]{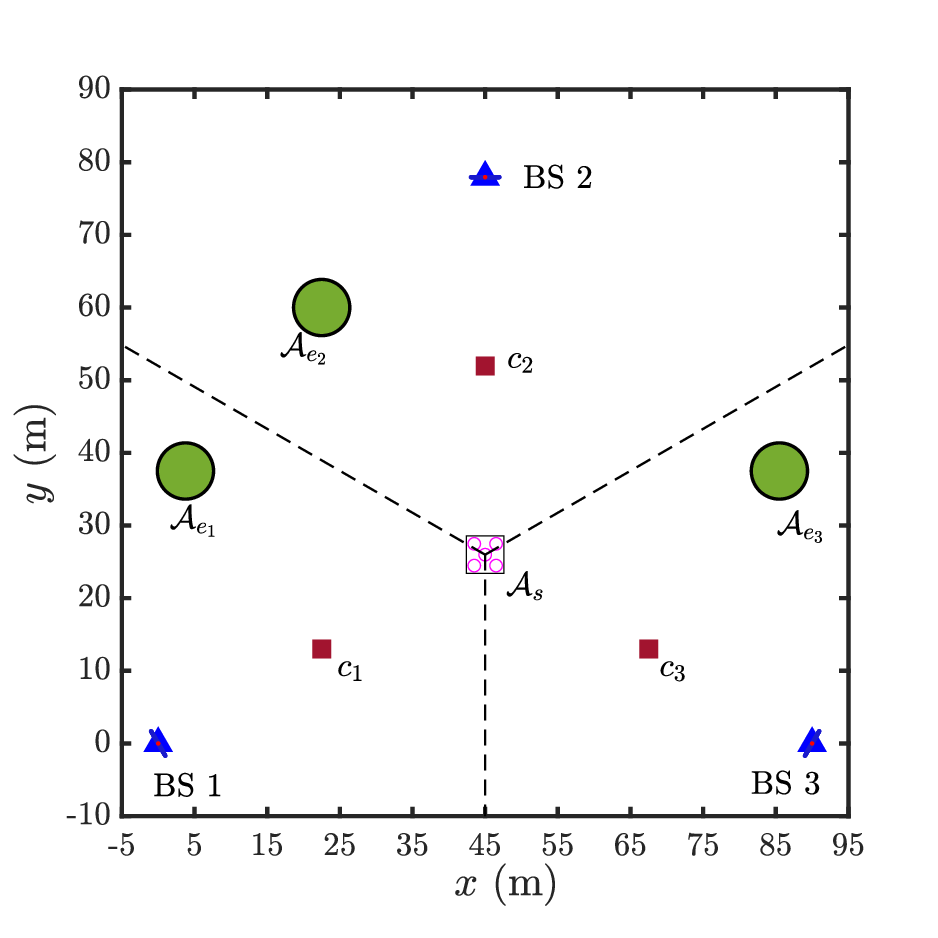}%
			\label{layout_case2}}
		\hfil
		\subfloat[Case~3.]{%
			\includegraphics[width=0.3\linewidth]{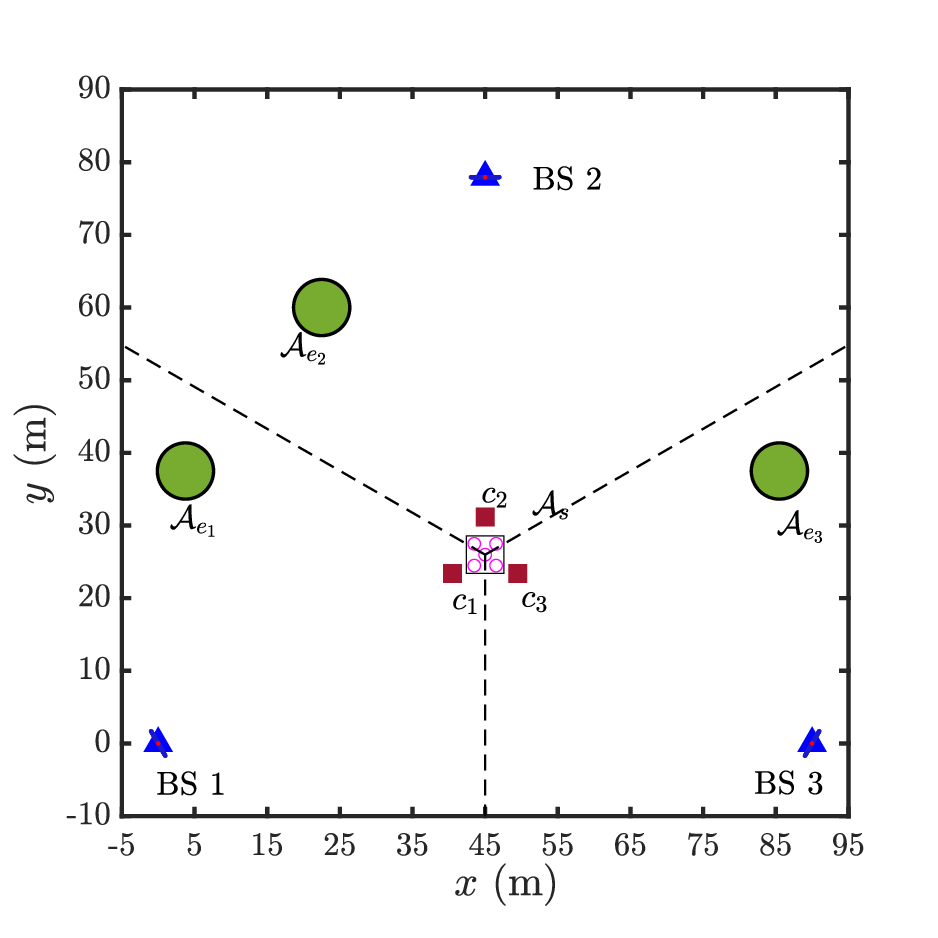}%
			\label{layout_case3}}
		\caption{Geometric layouts of the network showing BSs (blue triangles), CUs (red squares), ER uncertainty regions (green disks), and sensing sample points (pink circles) in three representative configurations.  Case 1: CUs are far from the sensing target; Case 2: CUs are at moderate distances from the sensing target; and Case 3: CUs are situated in close proximity to the sensing target.}
		\label{layoutlabel}
	\end{figure*}

	\section{Numerical Results}\label{numres}
	In this section, we present numerical results to validate the effectiveness of the proposed coordinated beamforming design for the multi-cell near-field ISCAP system. As illustrated in Fig.~\ref{layoutlabel}, we consider a network of $K=3$ BSs located at $(0,0)$\,m, $(45,45\sqrt{3})$\,m, and $(90,0)$\,m, with the corresponding Rayleigh distance $248.1$\,m. The antenna arrays at BS~1, BS~2, and BS~3 are oriented toward \(120^\circ\), \(0^\circ\), and \(60^\circ\), respectively. For the considered transmission interval, each ER is assumed to lie within a circular uncertainty region centered at \((3.75,37.5)\)\,m, \((22.5,60.0)\)\,m, and \((85.5,37.5)\)\,m, respectively. The sensing area is modeled as a square region with a side length of $3$\,m, where \( M = 5 \) sensing points are symmetrically placed within this square \cite{cheng2024optimal}. Moreover, throughout the numerical results, we consider $\Gamma_{c_i}=\Gamma_{c_j}$, $\forall c_i,c_j\in\mathcal{C}$, and $\Omega_{e_i}=\Omega_{e_j}$, $|\mathcal{A}_{e_i}|=|\mathcal{A}_{e_j}|$, $\forall e_i,e_j\in\mathcal{E}$. Unless otherwise stated, the rest of the system parameters are summarized in Table~\ref{tabparameters}. \textcolor{black}{It is important to note that, the developed mathematical
	framework are applicable for various network parameters, and
	the selection of these parameter values is for the purpose of
	presenting the achieved performance of our proposed beamforming
	designs.} Using different values will lead to a shifted
	network performance, but with the same conclusions. For comparison purposes, we evaluate the following different transmission configurations under various network scenarios,
	\begin{itemize}
		\item \textbf{SDR-based optimal:} the proposed SDR-based optimal beamforming design;
		\item \textbf{MRT-based sub-optimal:} the proposed MRT-based sub-optimal beamforming design;
		\item \textbf{Non-coordinated scheme:} a non-coordinated SDR-based beamforming design, where each BS optimizes its transmission independently without inter-cell coordination \cite{yilongiscap};
		\item \textbf{Worst-case robust:} an SDR-based beamforming design where the ER uncertainty region is discretized into $9$ sample points, and the harvested energy constraint is enforced at each point to ensure the minimum harvested power to meet the threshold \cite{Xiang6211384};
		\item \textbf{Far-field configuration:} obtained by setting a smaller number of antennas (e.g., $N = 16$) such that the array mainly operates in the far-field region \cite{cheng2024optimal}.
	\end{itemize}

	\begin{figure}[t]
		\centering
		\includegraphics[width=7cm]{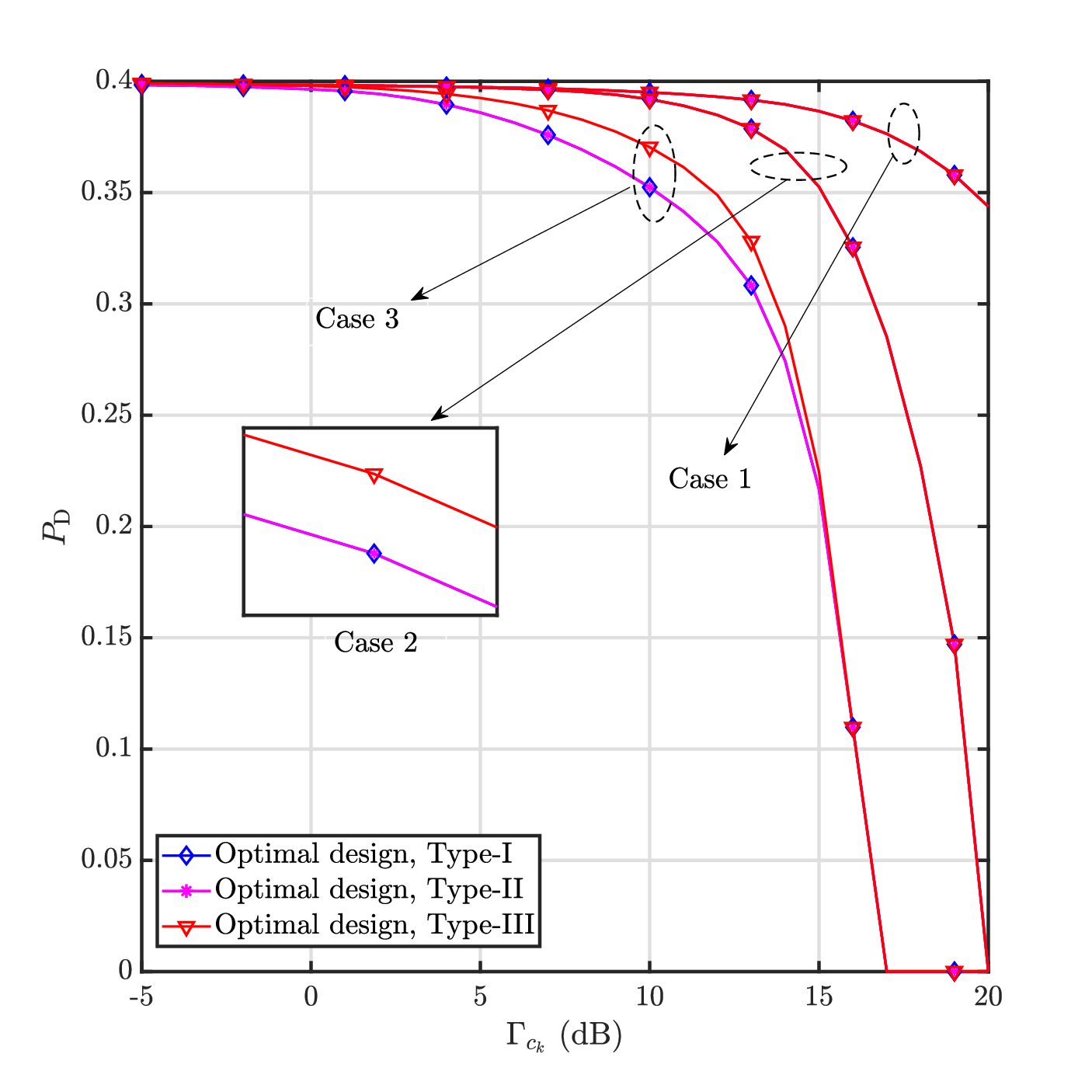}
		\caption{Detection probability $P_{\rm D}$ versus the CU SINR threshold $\Gamma_{c_k}$ (dB) under the SDR-based optimal beamforming design, where \( \Omega_{e_k} = -30\,\mathrm{dBm} \), \( P_{\rm FA} = 10^{-4} \), and $|\mathcal{A}_{e_k}|=0$.
		}
		\label{figlayoutdiscuss}
	\end{figure}
	
	\begin{figure}[t]
		\centering
		\includegraphics[width=7cm]{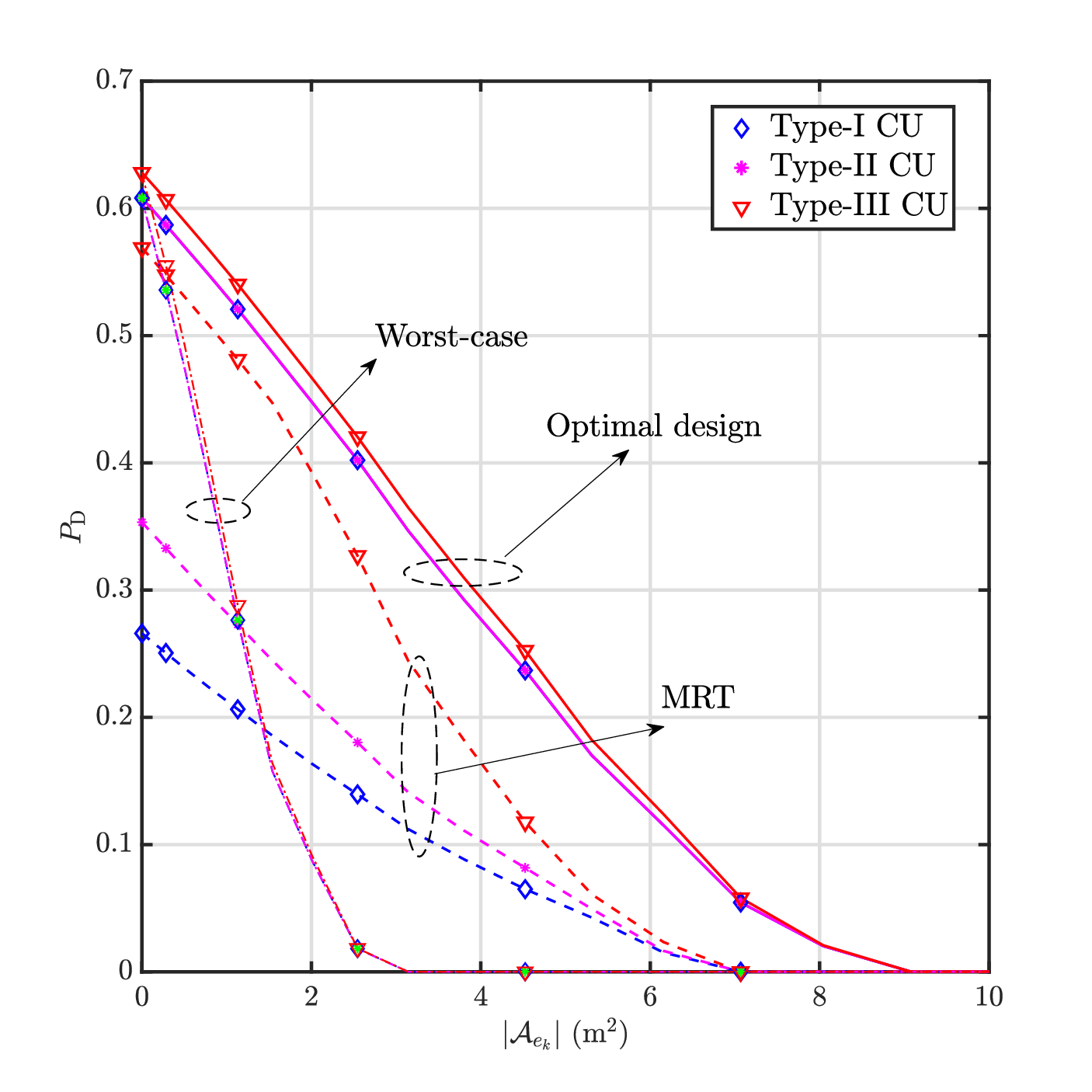}
		\caption{Detection probability \( P_{\rm D} \) versus area of uncertainty region $|\mathcal{A}_{e_k}|$ $(\mathrm{m}^2)$ for different CU types, beamforming designs, where \( \Omega_{e_k} = -30\,\mathrm{dBm} \), \( P_{\rm FA} = 10^{-4} \), and Case 3 configuration is adopted.}
		\label{area_effect}
	\end{figure}
	\begin{figure}[t]
			\centering
		\includegraphics[width=7cm]{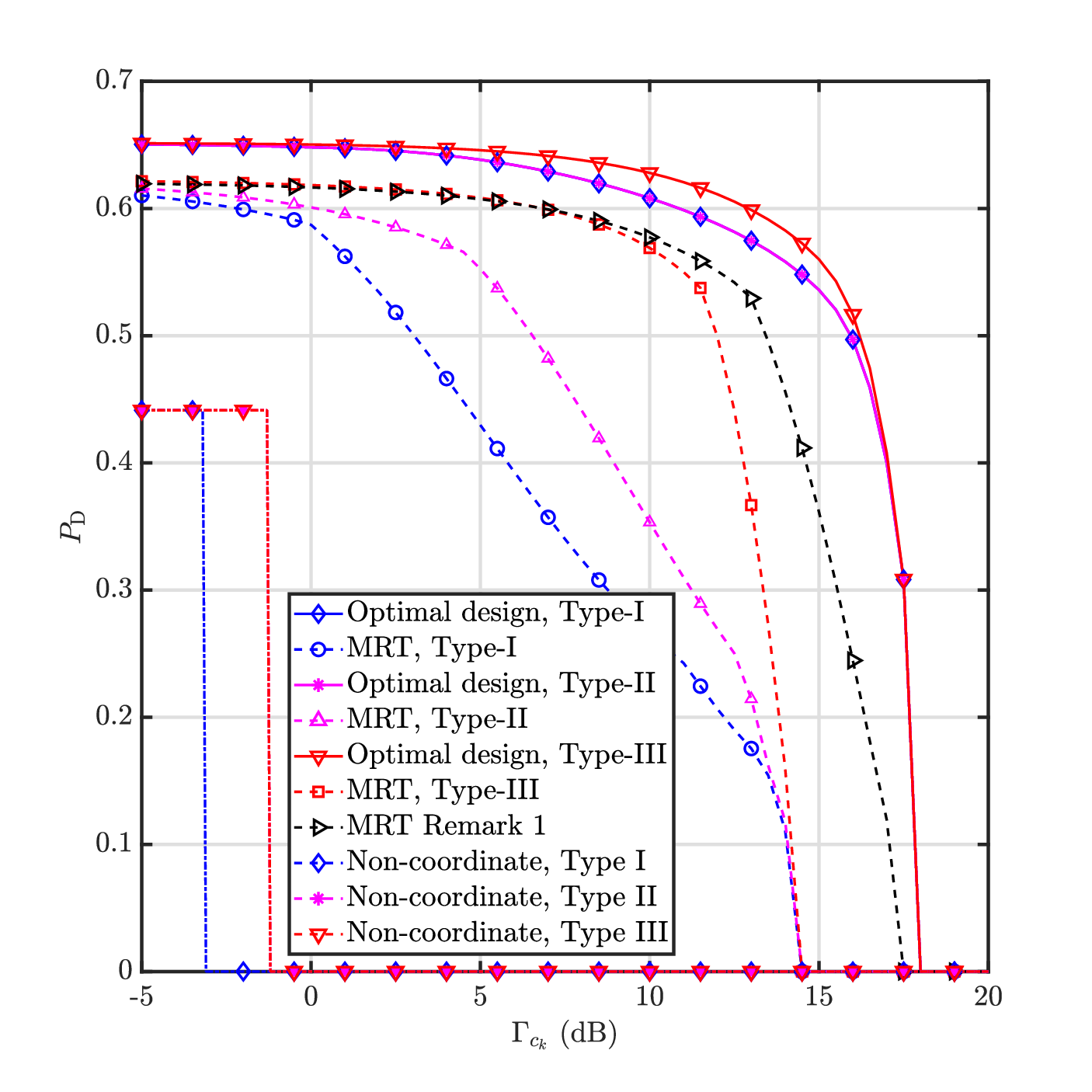}
		\caption{Detection probability \( P_{\rm D} \) versus SINR threshold $\Gamma_{c_k}$ (dB) for different CU types, beamforming designs, where \( \Omega_{e_k} = -30~ \mathrm{dBm} \), \( P_{\rm FA} = 10^{-4} \),  $|\mathcal{A}_{e_k}|=0$,  and Case 3 configuration is adopted.}
		\label{figfigure1}
	\end{figure}

	Fig.~\ref{figlayoutdiscuss} plots the detection probability \(P_{\rm D}\) versus the SINR threshold \(\Gamma_{c_k}\) (dB) with the SDR-based optimal beamforming design for three types of CUs across the spatial cases illustrated in Fig.~\ref{layoutlabel}. As expected, in all cases \(P_{\rm D}\) decreases monotonically with \(\Gamma_{c_k}\), since higher SINR requirements demand greater transmit power for communication, thereby reducing the power available for both energy transfer and sensing. Moreover, as revealed in Corollary~\ref{corollary},  Type-I and Type-II CUs yield identical performance, indicating that the dual-purpose signals already efficiently suppress intra-cell interference and that additional receiver-side cancellation in Type-II CUs provides no further gain. Furthermore, in both Case~1 and Case~2, where the CUs are located far from the sensing area, all CU types achieve nearly identical detection probability over the considered SINR regimes. This outcome indicates that, in such layouts, the proposed transmitter-side optimization sufficiently mitigates interference such that additional receiver-side signal processing only offers little benefit. By contrast, in Case~3, where the CUs are positioned close to the sensing area, Type-III achieves a clear performance advantage, showing that advanced receiver-side interference cancellation becomes beneficial once strong spatial coupling arises between communication and sensing.

	Fig.~\ref{area_effect} illustrates the impact of the ER uncertainty region size $|\mathcal{A}_{e_k}|$  on the sensing performance for different CU types and beamforming designs in Case~3 configuration. As expected, $P_{\rm D}$ decreases as $|\mathcal{A}_{e_k}|$ grows, since a larger uncertainty region forces the BS transmitter to spread energy more broadly to satisfy powering constraints across all possible ER locations, thereby reducing the energy that can be steered toward the intended sensing directions. \textcolor{black}{This trend also reflects the impact of ER location uncertainty on the effective CSI available for beam design, since a larger uncertainty region yields less accurate spatial information for shaping the transmit beams.} Furthermore, the proposed SDR-based optimal design achieves the best performance across all scenarios, benefiting from its joint optimization of beamforming directions and power allocation while explicitly accounting for interference, coupling effects, and spatial uncertainty. In contrast, the MRT-based scheme exhibits lower performance due to its fixed beam directions and limited flexibility in balancing sensing and power transfer. For comparison purposes, we also include the worst\mbox{-}case robust benchmark. When $|\mathcal{A}_{e_k}|=0$, it coincides with the proposed optimal design because the ER position is exactly known and the two formulations are indeed identical. As $|\mathcal{A}_{e_k}|$ increases, its performance degrades significantly since enforcing the harvesting constraint at all discretized locations effectively treats the uncertainty region as multiple possible ER positions that must all be satisfied, which yields a smaller feasible solution set, leading to conservative beam patterns and reduced sensing power.

	Fig.~\ref{figfigure1} plots the detection probability \( P_{\rm D} \) versus the SINR threshold \( \Gamma_{c_k} \) dB for the proposed SDR-based optimal design and the MRT-based sub-optimal design under various CU types. For comparison, a non-coordinated benchmark scheme is also included. Nevertheless, consistent with the distributed sensing system, the total echo power at each sensing point is still obtained by coherently combining the contributions from all BSs. Among the evaluated schemes, the SDR-based design achieves the best performance across all SINR thresholds. This is because SDR jointly optimizes the beamforming directions and power allocation of both the information and dual-purpose signals across BSs, allowing it to satisfy tighter SINR constraints with minimal impact on the energy steered toward the sensing region. In contrast, MRT adopts fixed beam directions aligned with CU channels; thus, as the SINR threshold increases, more transmit power must be diverted to the information beams, leaving less available for sensing. The asymptotic MRT approximation closely matches the actual MRT performance at low SINR thresholds but becomes an upper bound as the threshold increases. At low SINR, the required communication power is small and the assumptions of orthogonal channels and decoupled power allocation hold well. However, as the SINR threshold increases, the approximation overestimates the sensing power since it neglects interference and coupling effects, thereby providing a valid but optimistic upper bound on performance. Finally, both the proposed optimal and sub-optimal designs significantly outperform the non-coordinated benchmark, as inter-BS coordination enables effective mitigation of inter-cell interference, which in turn preserves more transmit power for sensing.

	Fig.~\ref{powereffect} shows the detection probability \( P_{\rm D} \) versus the powering threshold \( \Omega_{e_k} \) dBm for the proposed SDR-based optimal and MRT-based sub-optimal designs, by considering different CU types and a fixed SINR threshold of \( \Gamma_{c_k} = 10\,\mathrm{dB} \).  Firstly, it is observed that as \( \Omega_{e_k} \) increases, the detection probability gradually decreases across all schemes, since more transmit power must be allocated to energy transfer, reducing the power budget available for communication and sensing. Similar to the SINR-limited case in Fig.~\ref{figfigure1}, the proposed SDR-based optimal designs consistently outperform the MRT-based designs, and Type-III CUs achieve the best performance due to their advanced interference cancellation capability. 
	
	\begin{figure}[t]
		\centering
		\includegraphics[width=7cm]{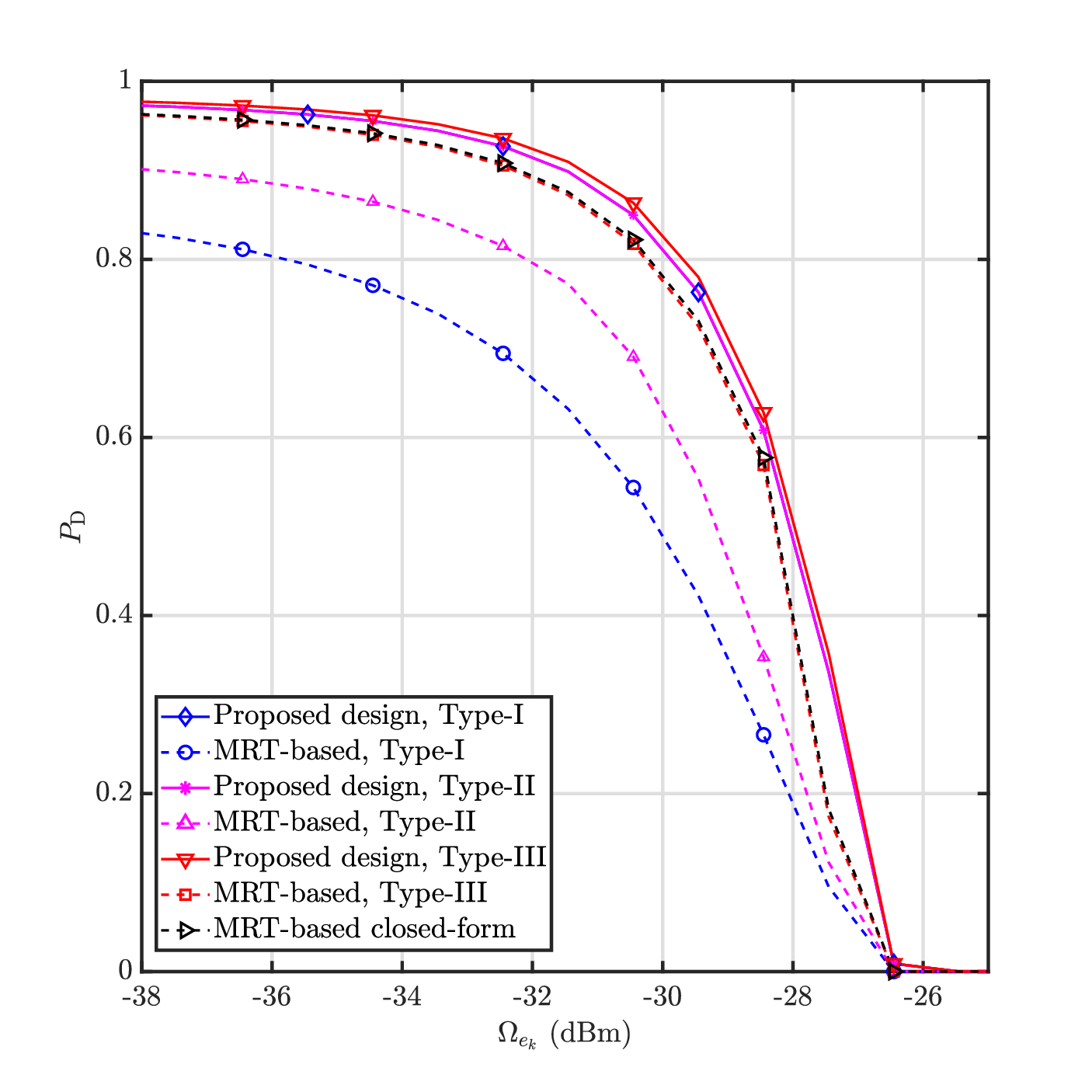}
		\caption{Detection probability \( P_{\rm D} \) versus powering threshold \( \Omega_{e_k} \) dBm for different CU types, beamforming designs, and number of antenna elements, where \( \Gamma_{c_k} = 10\,\mathrm{dB} \), \( |\mathcal{A}_{e_k}| = 0 \), and \( P_{\rm FA} = 10^{-4} \), and Case 3 configuration is adopted.}
		\label{powereffect}
	\end{figure}
	\begin{figure}[t]
			\centering
		\includegraphics[width=7cm]{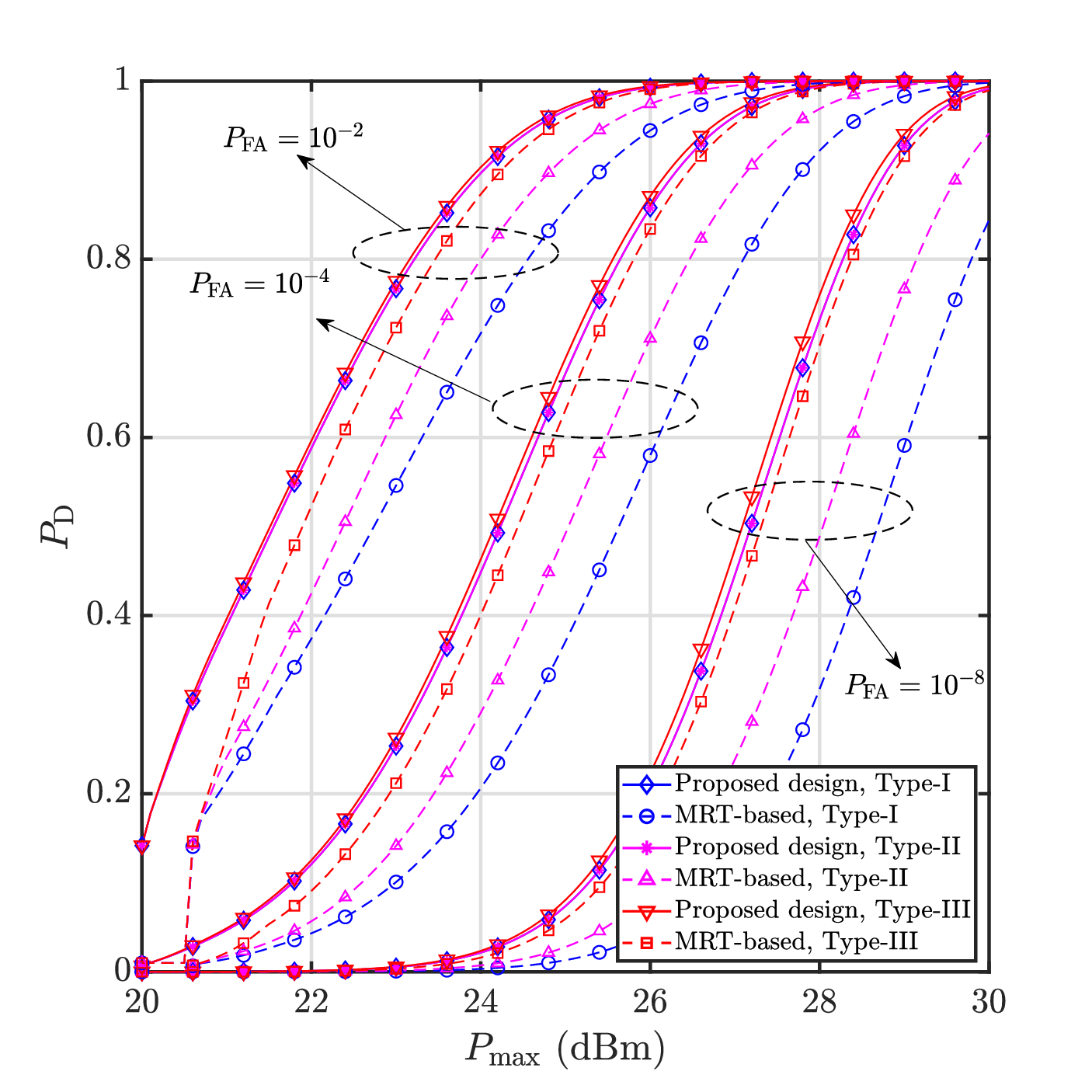}
		\caption{Detection probability \( P_{\rm D} \) versus transmit power budget for different CU types, beamforming designs and false alarm probabilities, where \( \Gamma_{c_k} = 10\,\mathrm{dB} \)  \(\forall c_k \in \mathcal{C} \),  \( \Omega_{e_k} = -36.55\,\mathrm{dBm} \), \( |\mathcal{A}_{e_k}| = 0 \), and Case 3 configuration is adopted. }
		\label{figfigure6}
	\end{figure}

	\begin{figure}[t]
		\centering
		\subfloat[SDR.]{%
			\includegraphics[width=0.5\linewidth]{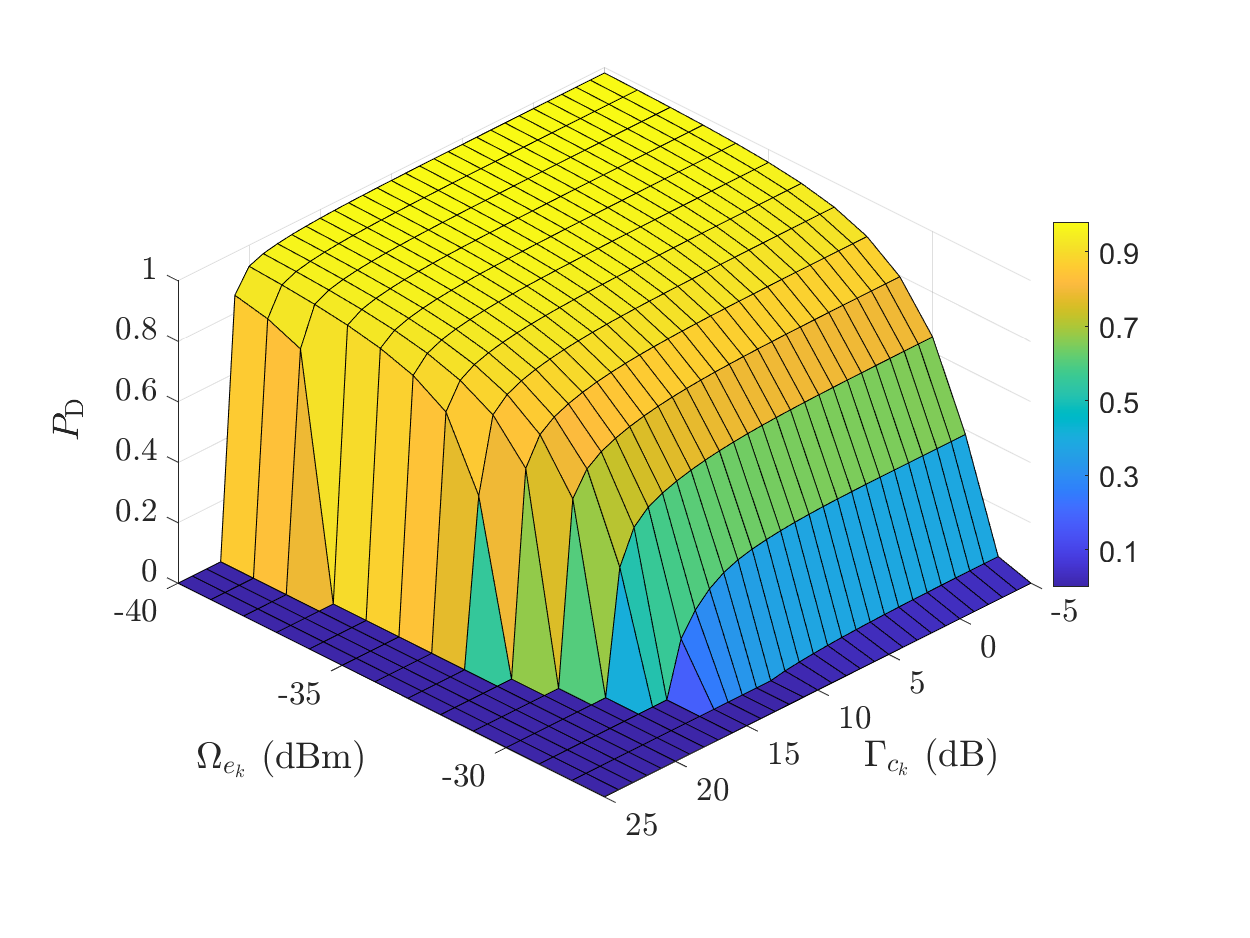}%
			\label{layout_case11}}
		\subfloat[MRT.]{%
			\includegraphics[width=0.5\linewidth]{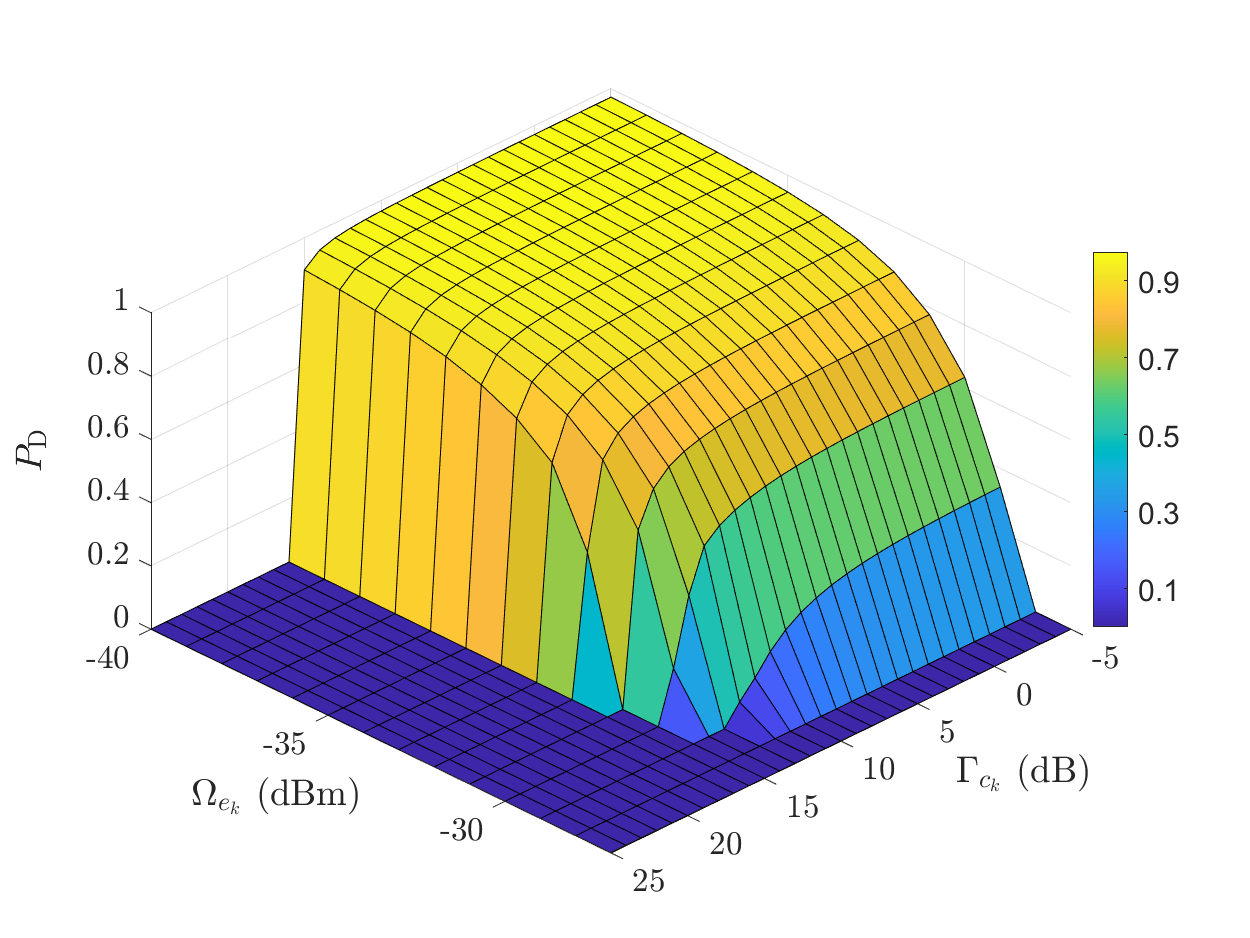}%
			\label{layout_case12}}
		\caption{Detection probability \( P_{\rm D} \) versus powering threshold \( \Omega_{e_k} \) dBm and SINR threshold \( \Gamma_{c_k} \) dB for the proposed SDR- and MRT-based beamforming designs with Type-III CUs, where Case 3 configuration is adopted.}
		\vspace{-5mm}   
		\label{iscaptradeoff}
	\end{figure}
	
	Fig.~\ref{figfigure6} demonstrates the joint impact of the transmit power budget \(P_{\max}\) and the false alarm probability \(P_{\mathrm{FA}}\) on the detection probability \(P_{\rm D}\), across various combinations of beamforming strategies and CU types. Overall, increasing \(P_{\max}\) leads to a consistent improvement in detection accuracy, as higher transmit power allows more energy to be allocated for effective sensing. In contrast, decreasing \(P_{\mathrm{FA}}\) reduces the detection probability, since a stricter false alarm requirement makes it more difficult for the sensing process to distinguish true signals from noise.  Moreover, the MRT-based scheme with Type-III CUs achieves detection performance close to that of SDR when \(P_{\mathrm{FA}}\) is extremely small (e.g., \(10^{-8}\)) and \(P_{\max}\) is high. This is because under such stringent sensing conditions, the SDR design becomes increasingly constrained and its detection performance tends to saturate, while the strong interference suppression capability of Type-III CUs, combined with high transmit power, enables MRT to maintain competitive detection accuracy, thereby narrowing the performance gap.
	
	Fig.~\ref{iscaptradeoff} illustrates the three-dimensional surfaces characterizing the detection probability $P_{\rm D}$ with respect to the powering threshold $\Omega_{e_k}$ dBm and the SINR threshold $\Gamma_{c_k}$ dB under the SDR- and MRT-based beamforming designs. In light of earlier results demonstrating that Type-III CUs consistently achieve the best detection performance, we focus on this configuration to clearly illustrate the trade-offs among different functionalities in near-field ISCAP systems.  It is observed that the detection probability drops more quickly with increasing \( \Omega_{e_k} \) than with \( \Gamma_{c_k} \). This is because satisfying higher powering demands necessitates allocating more transmit power toward the ERs, thereby reducing the portion that can be directed to sensing. On the other hand, SINR constraints can typically be satisfied without severely affecting sensing, since the large antenna array enables high-resolution beamforming in the near-field regime to mitigate the interference. Finally, it is evident from the figures that the SDR-based design achieves superior detection performance to MRT, particularly under high SINR and powering thresholds, owing to its enhanced flexibility in beamforming and power allocation.
	\begin{figure*}[!t]
		\centering
		\subfloat[Received power from BS 1.]{%
			\includegraphics[width=0.33\linewidth]{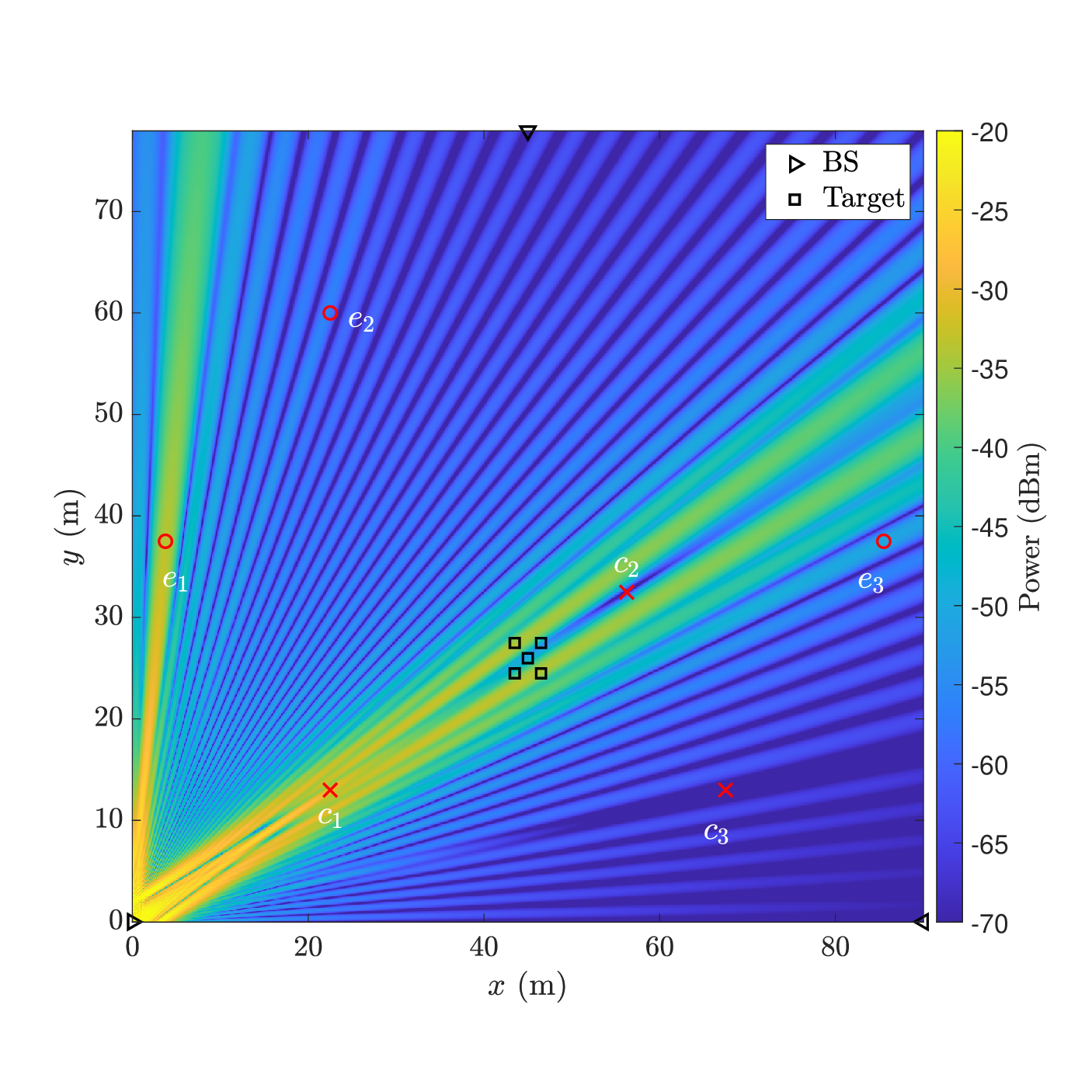}%
			\label{figa}}
		\hfil
		\subfloat[Received power from BS 2.]{%
			\includegraphics[width=0.33\linewidth]{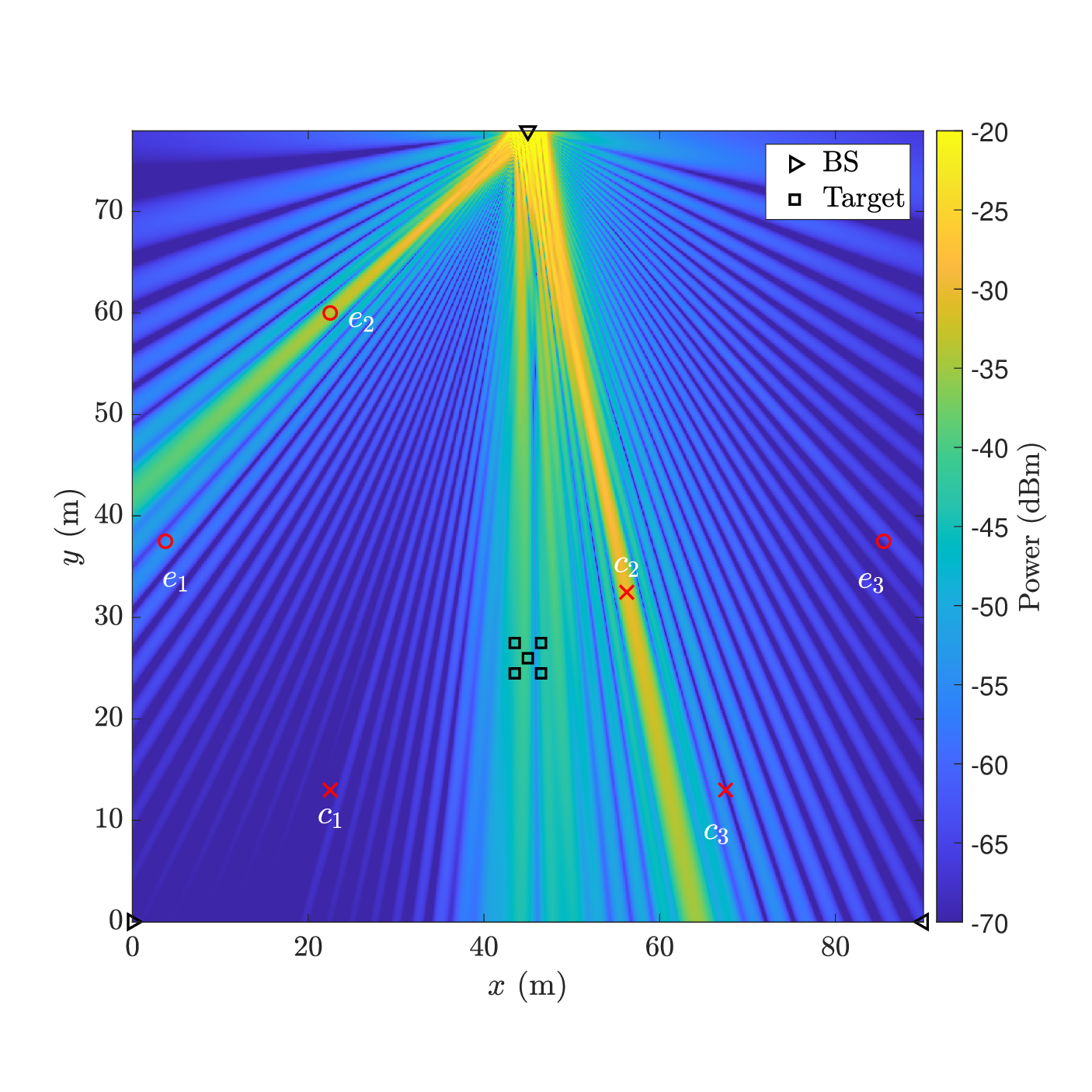}%
			\label{figb}}
		\hfil
		\subfloat[Received power from BS 3.]{%
			\includegraphics[width=0.33\linewidth]{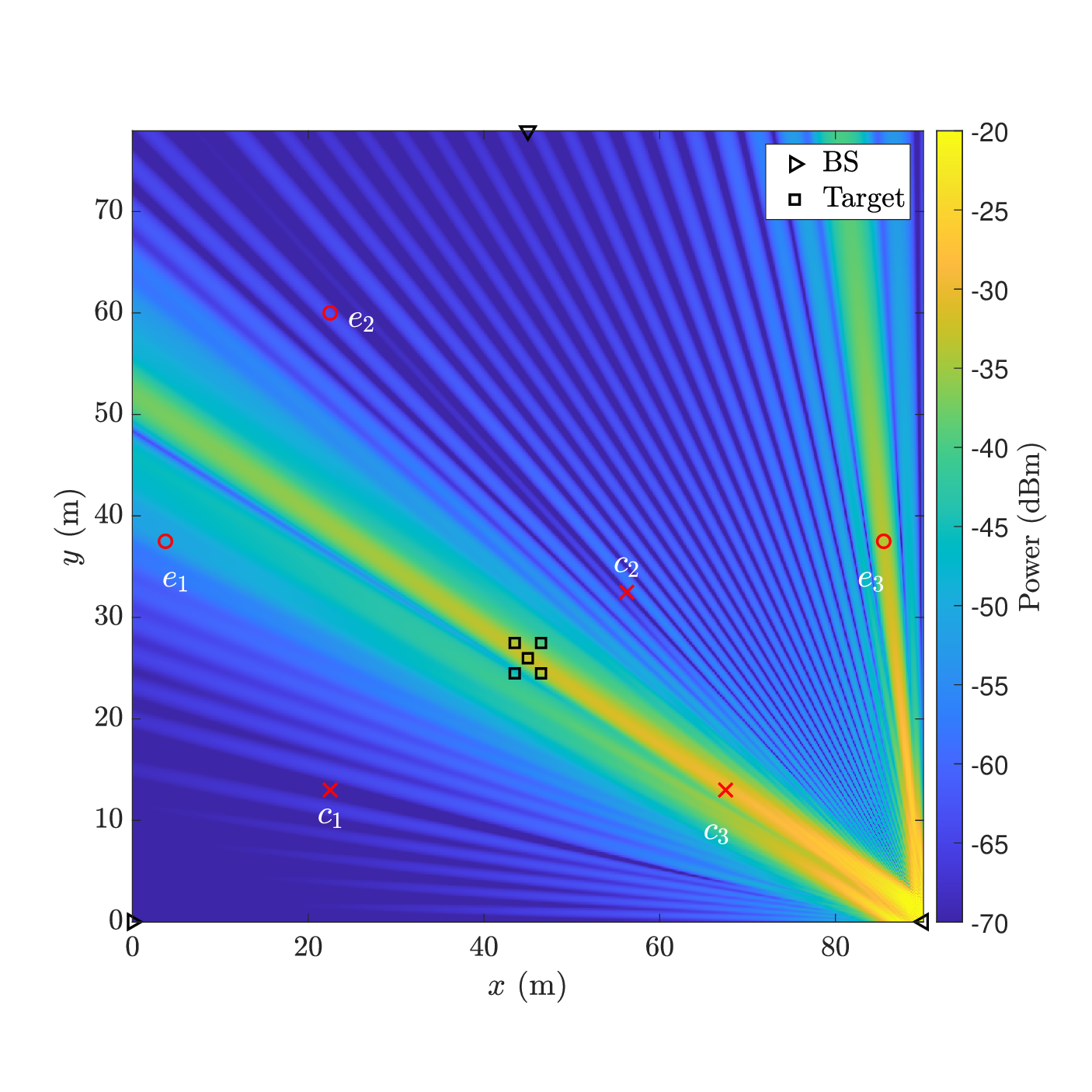}%
			\label{figc}}
		\vspace{-4mm}
		\subfloat[Received power from BS 1.]{%
			\includegraphics[width=0.33\linewidth]{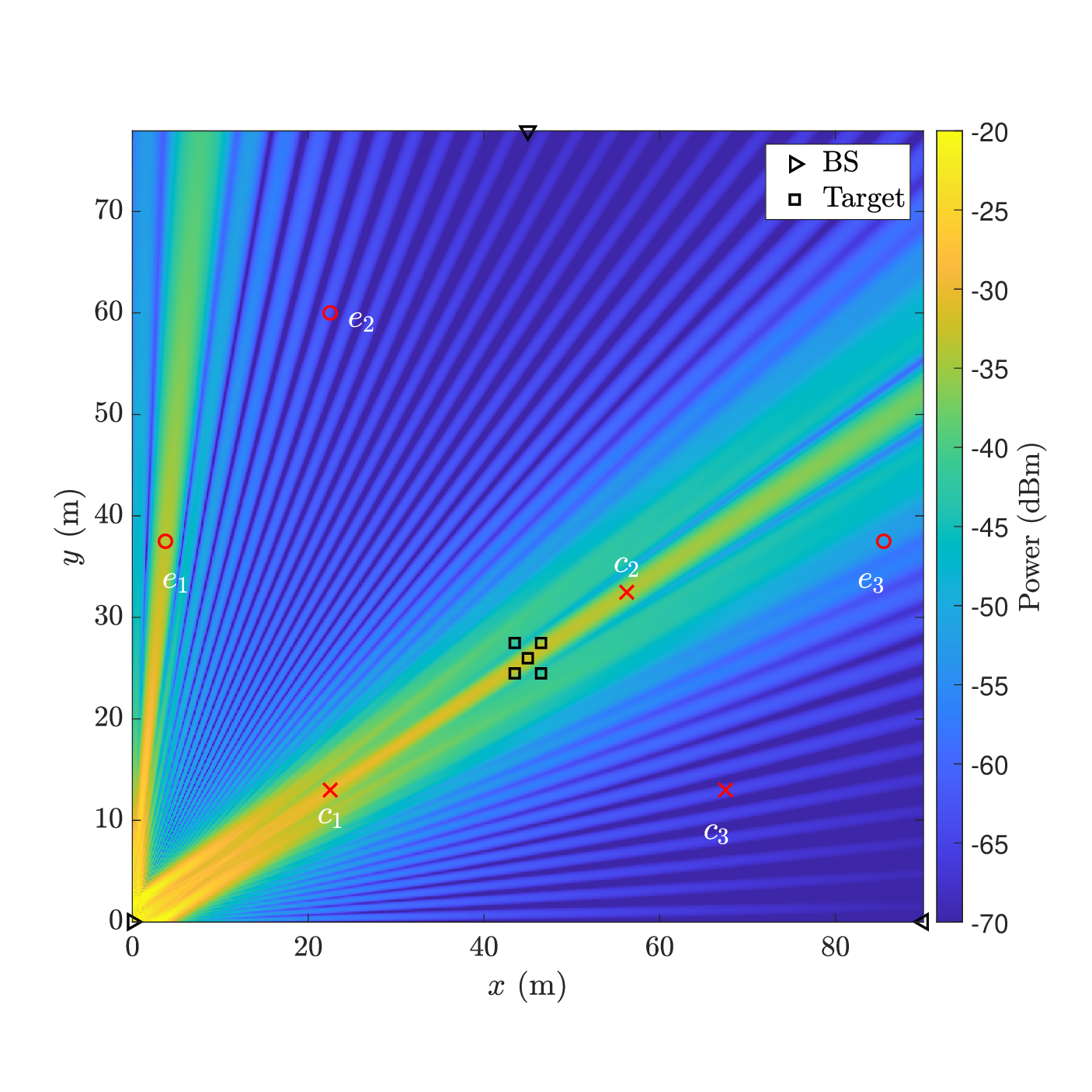}%
			\label{figd}}
		\hfil
		\subfloat[Received power from BS 2.]{%
			\includegraphics[width=0.33\linewidth]{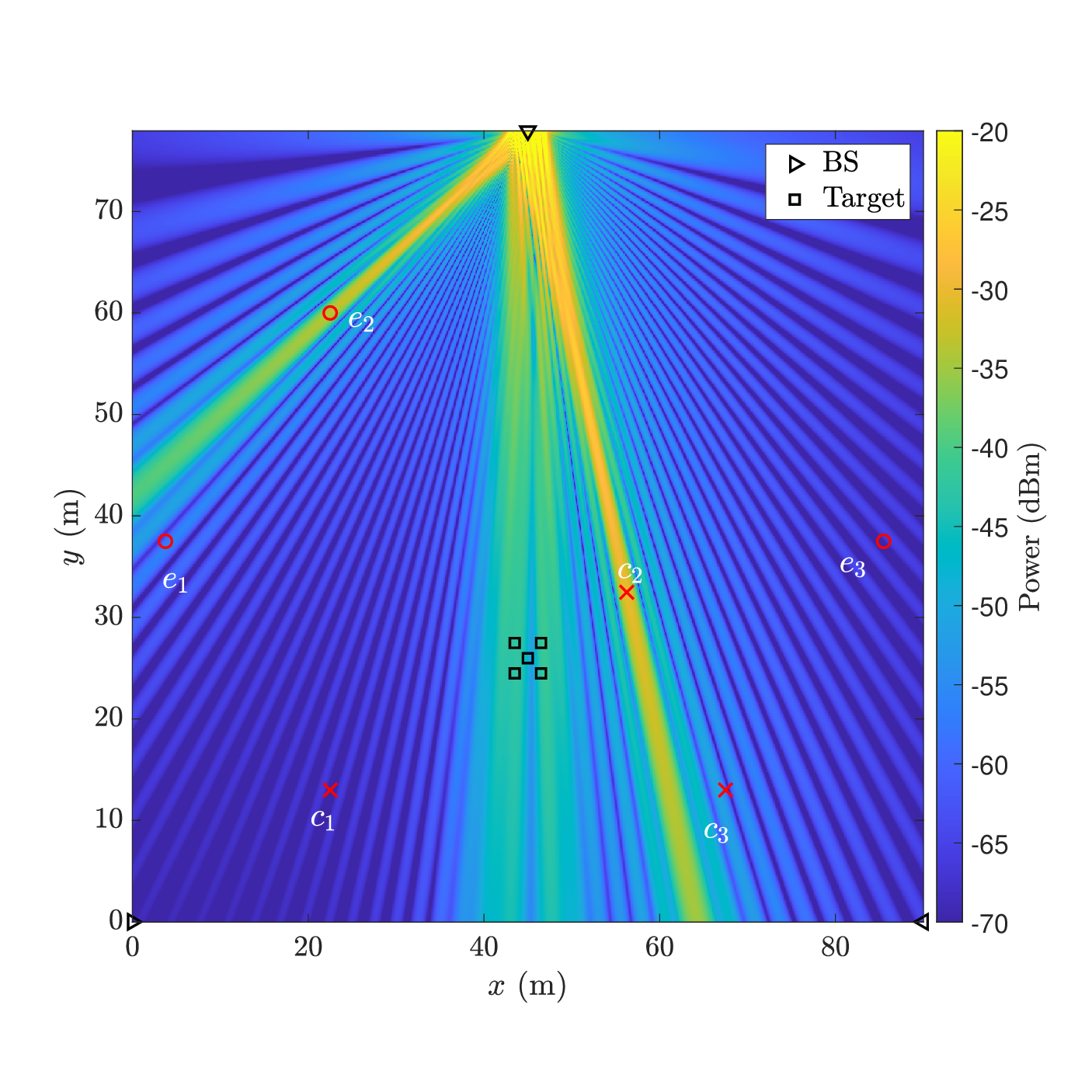}%
			\label{fige}}
		\hfil
		\subfloat[Received power from BS 3.]{%
			\includegraphics[width=0.33\linewidth]{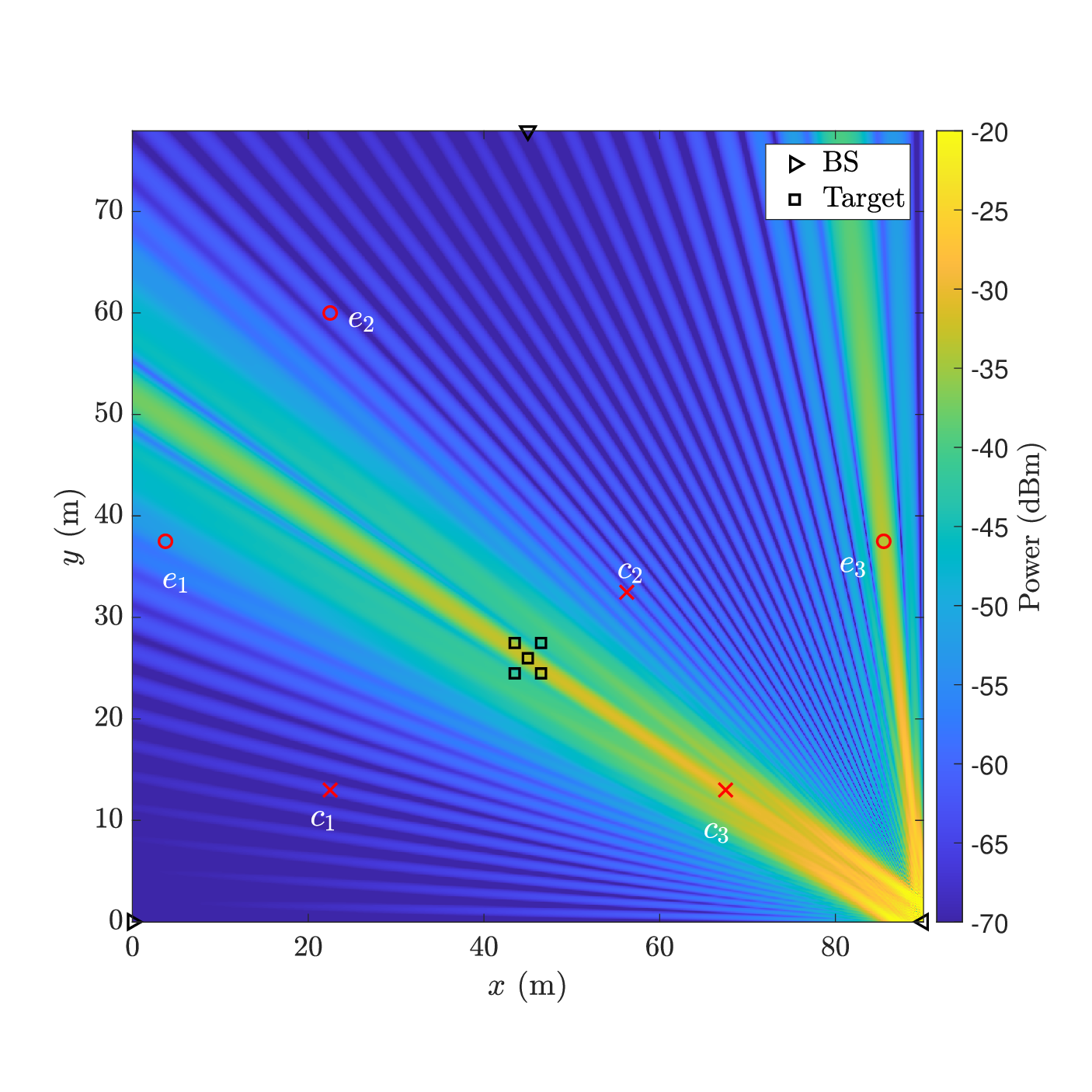}%
			\label{figf}}
		\caption{Received power maps from each BS (triangles) with (a–c) coordinated beamforming and (d–f) non-coordinated beamforming under a Type-I CU configuration, where $\Gamma_{c_k}=20$ dB, $\Omega_{e_k}=-35$ dBm, and $|\mathcal{A}_{e_k}|=0$.}
		\label{powermap}
	\end{figure*}
	\begin{figure*}[!t]
		\centering
		\subfloat[Received power from BS 1.]{%
			\includegraphics[width=0.33\linewidth]{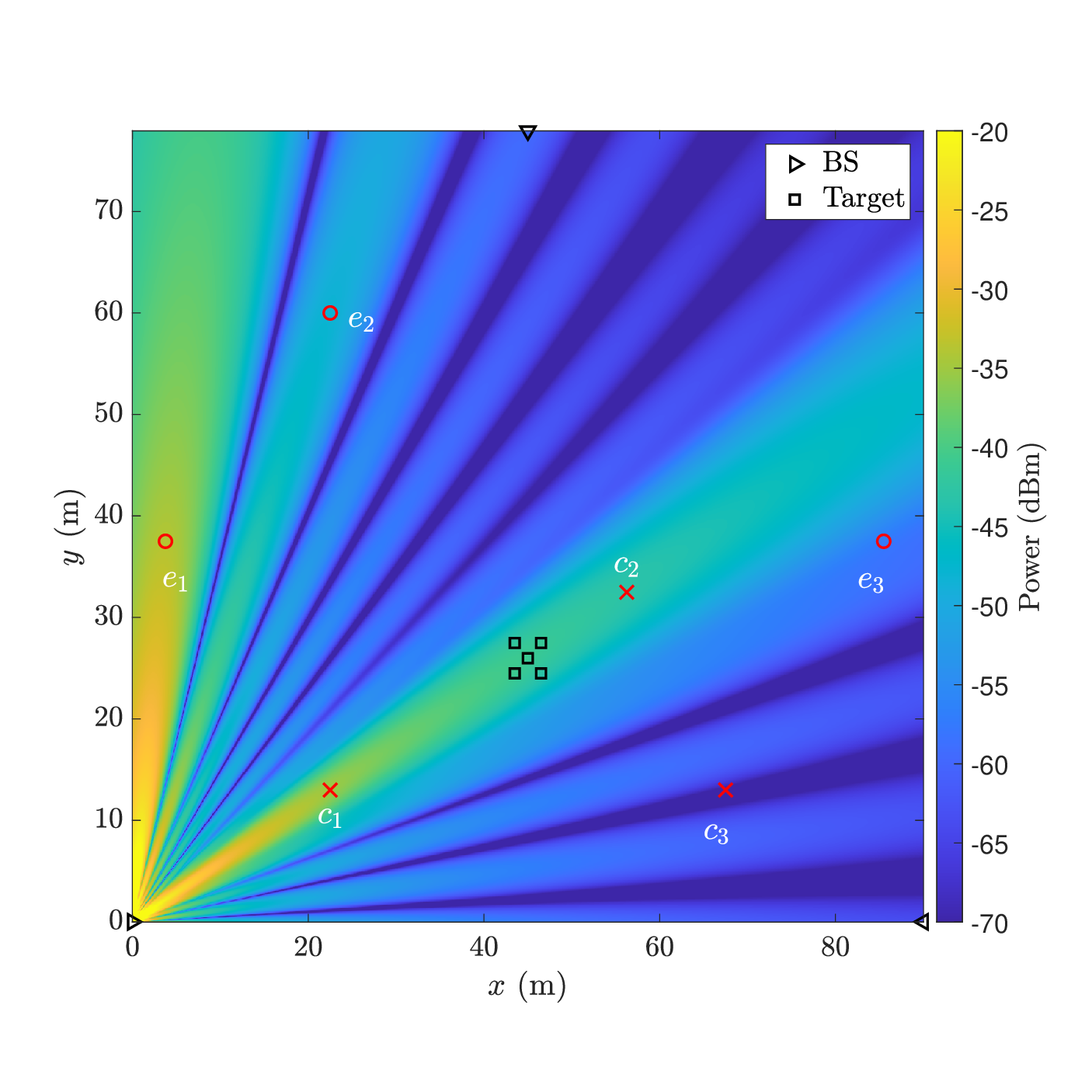}%
			\label{layout_case13}}
		\hfil
		\subfloat[Received power from BS 2.]{%
			\includegraphics[width=0.33\linewidth]{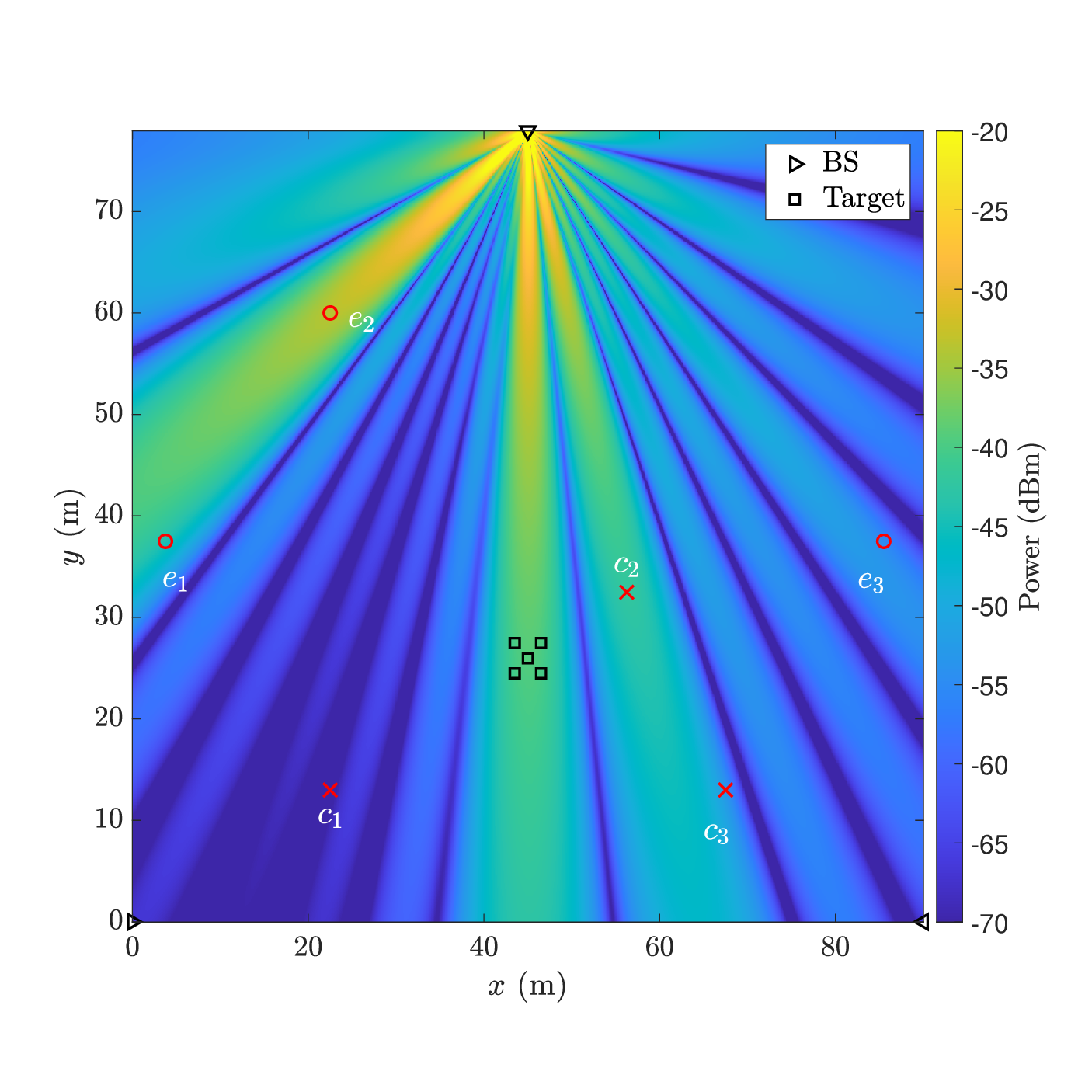}%
			\label{layout_case23}}
		\hfil
		\subfloat[Received power from BS 3.]{%
			\includegraphics[width=0.33\linewidth]{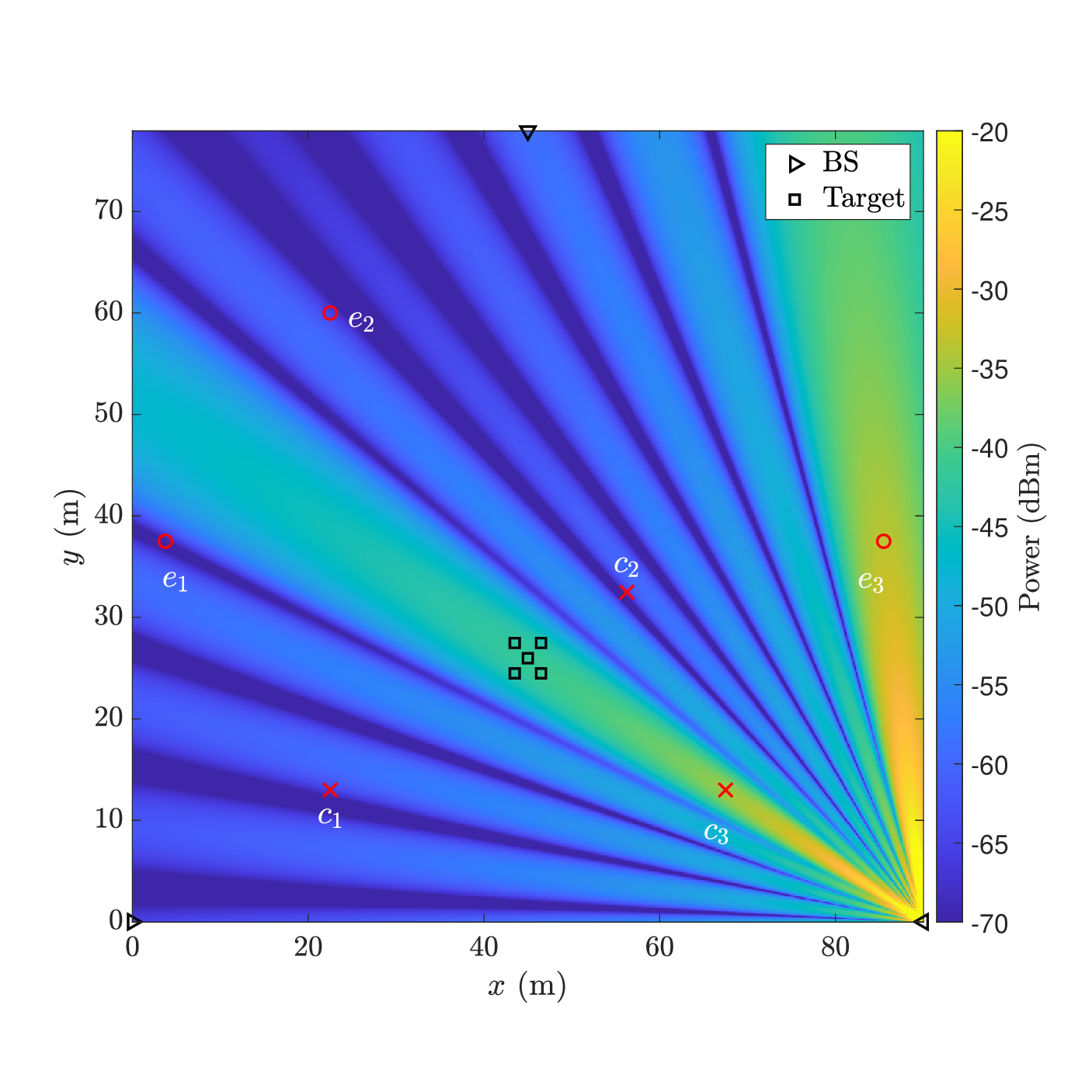}%
			\label{layout_case33}}
		\caption{Received power maps from each BS (triangles) with coordinated beamforming under a Type-I CU configuration, where $N=16$, $\Gamma_{c_k}=0$ dB, $\Omega_{e_k}=-35$ dBm, and $|\mathcal{A}_{e_k}|=0$.}
		\label{farpowermap}
	\end{figure*}
	
	Fig.~\ref{powermap} presents the received power maps from each BS under two configurations, namely, coordinated beamforming in Fig.~\ref{powermap}\subref{figa}-\subref{figc}, and non-coordinated beamforming in Fig.~\ref{powermap}\subref{figd}-\subref{figf}. It is observed that under coordinated beamforming, each BS focuses energy toward its intended CU while forming deep nulls at unintended CU locations. These nulls are clearly visible as dark regions around CU markers in Fig.~\ref{powermap}\subref{figa}-\ref{powermap}\subref{figc}, reflecting effective interference suppression. A particularly notable case occurs at BS~1, where CUs \(c_1\) and \(c_2\) are nearly collinear with the BS. Despite this alignment, BS~1 directs a power focus toward \(c_1\) while simultaneously nulling interference at \(c_2\), as observed in Fig.~\ref{powermap}\subref{figa}. This demonstrates the fine angular and range control achievable through near-field coordination. In contrast, the non-coordinated case in Fig.~\ref{powermap}\subref{figd}-\subref{figf} shows significant power leakage at unintended CU locations, leading to strong inter-cell interference.

	Fig.~\ref{farpowermap} presents the received power distributions from each BS under coordinated transmission in a far-field configuration, where the number of antennas is reduced to \(N = 16\), yielding a Rayleigh distance of approximately \(14\) m. \textcolor{black}{Note that far-field propagation corresponds to the special case of the spherical-wave model when the array aperture is small relative to the link distance (i.e., operation beyond the Rayleigh distance); under this condition, the adopted near-field channel remains valid and naturally approximates the conventional planar-wave representation \cite{liu10716601}.} In this regime, beamforming primarily operates in the angular domain and lacks range resolution, which makes it challenging to spatially separate CUs that lie along similar angular directions. Compared with the near-field case shown in Fig.~\ref{powermap}, the power maps in Fig.~\ref{farpowermap}\subref{layout_case13}--\subref{layout_case33} reveal broader main lobes and a diminished ability to form deep nulls toward unintended CU locations. For instance, although BS~1 attempts to focus energy toward CU~\(c_1\), a substantial portion still leaks toward CU~\(c_2\), even under coordinated design. Similar leakage patterns are observed for BS~2 and BS~3. These results reveal that, despite inter-cell coordination, far-field beamforming remains fundamentally limited in spatial resolution and cannot simultaneously achieve sharp energy focusing and effective interference suppression, especially when multiple CUs appear along nearly the same angular direction from the perspective of the transmitting BSs. \textcolor{black}{In particular, the configuration with $N=16$ highlights the impact of a limited array aperture, where the reduced spatial degrees of freedom constrain the energy-focusing capability and result in noticeable performance degradation relative to larger-array configurations.}

	\section{Conclusion}\label{conlabel}
	This paper investigated the coordinated multi-cell ISCAP in the electromagnetic near field exploiting ELAAs, with a particular focus on the practical challenge of ER location uncertainty. A unified optimization framework was developed to jointly design the information and dual-purpose signals to maximize worst-case sensing performance, subject to SINR requirements at various types of CUs and average power delivery constraints at ERs over uncertain regions. By leveraging SDR, the inherently non-convex problems were transformed into convex form with guaranteed global optimality, while a low-complexity MRT-based design was also proposed.	Numerical results demonstrated that inter-BS coordination is crucial in harnessing inter-cell interference, which in turn preserves more transmit power for sensing and powering. The proposed SDR-based design consistently outperformed the MRT-based and non-coordinated schemes across a wide range of SINR thresholds, power demands, and system settings. The results also reveal fundamental trade-offs among sensing accuracy, communication reliability, and WPT efficiency. Additionally, it was shown that detection probability increases with larger BS power budgets, but is limited by stricter false alarm constraints. Moreover, Type-I and Type-II CUs achieve identical performance, while Type-III CUs lead to extra gains only with strong communication-sensing coupling. The proposed SDR-based optimal scheme also outperforms worst-case robust benchmarks and far-field configurations, highlighting the superior spatial resolution and interference management of near-field ISCAP system. 
	
	\textcolor{black}{Future research may extend the proposed framework to dynamic scenarios involving highly mobile users and rapidly time-varying channels, where learning based adaptive optimization methods, including federated-learning-based distributed optimization and reinforcement learning, may offer improved adaptability and real-time decision-making capabilities. Another promising direction is to explicitly incorporate the effects of imperfect CSI for both CUs and EUs, including estimation errors, feedback delays, and model uncertainty, and to develop robust beamforming strategies that remain reliable under such conditions. Finally, hardware-aware optimization that accounts for a practical transceiver imperfections, finite-resolution phase control, mutual coupling, and calibration constraints, together with hybrid near- and far-field operation, represents an important avenue for further investigation.}

	\appendices

	\section{Proof of Proposition \ref{prop}} \label{proofofprop}   
	Since the objective function, the BS power budget constraint, and the energy harvesting constraint depend only on the total transmit covariance, it can be easily verified that the reconstructed solution achieves the same objective value as \(\mathrm{(SDR1.1)}\), and satisfies the power budget constraint at each BS as well as the energy harvesting constraint at each ER. In addition, the reconstructed matrix \(\widetilde{\bm{\mathcal{W}}}_k = \widetilde{\bm{\omega}}_k \widetilde{\bm{\omega}}_k^H\) is rank-one for each \(k \in \{1,\dotsc,K\}\), since it is the outer product of a single vector. 
	
	We now show that the reconstructed matrices \(\{\widetilde{\bm{\mathcal{W}}}_k\}\) also satisfy the SINR constraints. From the definition of \(\widetilde{\bm{\omega}}_k\), we have $
	\bm{h}_{l,c_k}^H \widetilde{\bm{\mathcal{W}}}_l \bm{h}_{l,c_k}
	= \bm{h}_{l,c_k}^H \widetilde{\bm{\omega}}_l \widetilde{\bm{\omega}}_l^H \bm{h}_{l,c_k}
	= \bm{h}_{l,c_k}^H \bm{\mathcal{W}}_l^{\mathrm{I}} \bm{h}_{l,c_k},\quad \forall l \in \mathcal{K}$.    Thus, the SINR constraint for CU \(c_k\) becomes
	\begin{align}
		&\left(1 + \frac{1}{\Gamma_{c_k}}\right) \bm{h}_{k,c_k}^H \widetilde{\bm{\mathcal{W}}}_k \bm{h}_{k,c_k}
		= \left(1 + \frac{1}{\Gamma_{c_k}}\right) \bm{h}_{k,c_k}^H \bm{\mathcal{W}}_k^{\mathrm{I}} \bm{h}_{k,c_k} \nonumber \\
		\ge& \sum\nolimits_{l=1}^K \bm{h}_{l,c_k}^H \left( \bm{\mathcal{W}}_l^{\mathrm{I}} + \bm{R}_{\mathrm{o},l}^{\mathrm{I}} \right) \bm{h}_{l,c_k} + \sigma_{\mathrm{c}}^2 \nonumber \\
		=& \sum\nolimits_{l=1}^K \bm{h}_{l,c_k}^H \left( \widetilde{\bm{\mathcal{W}}}_l + \widetilde{\bm{R}}_{\mathrm{o},l} \right) \bm{h}_{l,c_k} + \sigma_{\mathrm{c}}^2.
	\end{align}
	Finally, we verify that \(\widetilde{\bm{R}}_{\mathrm{o},k} = \bm{\mathcal{W}}_k^{\mathrm{I}} + \bm{R}_{\mathrm{o},k}^{\mathrm{I}} - \widetilde{\bm{\mathcal{W}}}_k\) is positive semidefinite. Since both \(\bm{\mathcal{W}}_k^{\mathrm{I}}\) and \(\bm{R}_{\mathrm{o},k}^{\mathrm{I}}\) are positive semidefinite, their sum is also positive semidefinite. Moreover, \(\widetilde{\bm{\mathcal{W}}}_k\) is constructed from \(\bm{\mathcal{W}}_k^{\mathrm{I}}\) as a rank-one projection that satisfies \(\widetilde{\bm{\mathcal{W}}}_k \preceq \bm{\mathcal{W}}_k^{\mathrm{I}}\). Therefore, subtracting \(\widetilde{\bm{\mathcal{W}}}_k\) from the sum still leads to a positive semidefinite matrix, implying that \(\widetilde{\bm{R}}_{\mathrm{o},k} \succeq \bm{0}\).
	
	Hence, the reconstructed solution \(\{\widetilde{\bm{\mathcal{W}}}_k\}\) and \(\{\widetilde{\bm{R}}_{\mathrm{o},k}\}\) are feasible for problem \((\mathrm{P1.1})\), satisfying all constraints and preserving optimality. This completes the proof of Proposition~\ref{prop}.

	\section{Proof of Corollary \ref{corollary}}\label{corollaryproof}
	Let $\{\bm{\mathcal{W}}_k^{\rm I}, \bm{R}_{\mathrm{o},k}^{\rm I}\}$,  $\forall\,k\in\{1,...,K\}$ denote an optimal solution to \textnormal{(SDR1.1)}, achieving the objective value $\Theta^{\rm I}$. For the sake of contradiction, we suppose that there exists at least one BS $k$ such that $\bm{R}_{\mathrm{o},k}^{\rm I}$ is not orthogonal to the CU channel vector $\bm{h}_{k,c_k}$. Define the unit vector $\bm{u}_k = \bm{h}_{k,c_k} / \|\bm{h}_{k,c_k}\|$. Since $\bm{R}_{\mathrm{o},k}^{\rm I} \succeq \bm{0}$, we write $\bm{R}_{\mathrm{o},k}^{\rm I} = \bm{R}_{\perp,k} + \alpha_k \bm{u}_k \bm{u}_k^H$, where $\bm{R}_{\perp,k} \bm{u}_k = \bm{0}$ and $\alpha_k = \bm{u}_k^H \bm{R}_{\mathrm{o},k}^{\rm I} \bm{u}_k > 0$.
	
	Construct an alternative solution by defining $\widehat{\bm{\mathcal{W}}}_k = \bm{\mathcal{W}}_k^{\rm I} + \alpha_k \bm{u}_k \bm{u}_k^H$ and $\widehat{\bm{R}}_{\mathrm{o},k} = \bm{R}_{\perp,k}$. The total transmit covariance remains unchanged, i.e., $\widehat{\bm{\mathcal{W}}}_k + \widehat{\bm{R}}_{\mathrm{o},k} = \bm{\mathcal{W}}_k^{\rm I} + \bm{R}_{\mathrm{o},k}^{\rm I}$, so all original constraints are satisfied except the SINR constraint at CU $c_k$. The interference term from $\bm{R}_{\mathrm{o},k}^{\rm I}$ is eliminated and the desired signal power increases due to the added rank-one component in $\widehat{\bm{\mathcal{W}}}_k$ aligned with $\bm{h}_{k,c_k}$. Since interference from other BSs remains unchanged, the SINR at CU $c_k$ strictly increases. This implies that the SINR constraint becomes inactive under the constructed solution, contradicting the optimality of the original one, where all constraints must be active at optimum. Therefore, $\alpha_k = 0$ must hold for all CUs, implying that $\bm{h}_{k,c_k}^H \bm{R}_{\mathrm{o},k}^{\rm I} \bm{h}_{k,c_k} = 0$. At optimality, Type-I SINR constraints thus reduce to the Type-II counterparts, and since all remaining constraints are identical and active, it follows that $\Theta^{\rm I} = \Theta^{\rm II}$.

	\bibliographystyle{IEEEtran}
	\bibliography{ref}

\end{document}